\documentclass{article}
\usepackage{graphicx} 
\usepackage{parskip}
\usepackage{caption} 

\usepackage[title]{appendix}
\usepackage{apacite}
\usepackage{tabularray}
\usepackage[margin=0.9in]{geometry}
\usepackage{setspace}
\usepackage{amsmath,amssymb,amsthm,bm}
\usepackage{graphicx,psfrag,epsf}
\usepackage{multicol}
\usepackage{multirow}
\usepackage{mathtools,amsmath}
\usepackage{bbm}
\usepackage{flafter}
\usepackage{tfrupee}
\usepackage{eurosym}
\usepackage{comment}
\usepackage{color,hyperref,xcolor}
\hypersetup{colorlinks=true,urlcolor=purple,citecolor=blue}
\usepackage{cleveref}
\usepackage{longtable}
\usepackage{booktabs}
\usepackage{dcolumn}

\usepackage{tikz}
\usetikzlibrary{shapes.geometric, arrows}
\tikzstyle{startstop} = [rectangle, rounded corners, minimum width=3cm, minimum height=1cm,text centered, draw=black]
\tikzstyle{io} = [trapezium, trapezium left angle=70, trapezium right angle=110, minimum width=3cm, minimum height=1cm, text centered, draw=black, fill=blue!30]
\tikzstyle{dot}=[\ldots]
\tikzstyle{arrow} = [thick,->,>=stealth]

\newcolumntype{L}{D{.}{.}{2,12}}
\usepackage{siunitx} 
\usepackage{enumitem}
\usepackage{natbib}
\usepackage{url} 
\usepackage[ruled,vlined]{algorithm2e}
\usepackage[flushleft]{threeparttable} 
\usepackage{booktabs,caption}
\numberwithin{equation}{section}
\usepackage[T1]{fontenc}
\usepackage{diagbox}
\usepackage{siunitx,booktabs}
\usepackage{tabularx}

\def\R{\mathbb{R}}

\def\bY{\bm{Y}}
\def\bX{\bm{X}}
\def\bC{\bm{C}}

\renewcommand{\vec}[1]{\mathbf{1}}
\DeclareMathOperator*{\argmin}{arg\min}
\renewcommand{\P}{\mathrm{P}}
\newcommand{\E}{\mathbb{E}}
\newcommand{\prob}{\mathbb{P}}
\newcommand{\btheta}{\bm{\theta}}

\newcommand{\normal}{\mathrm{N}}

\newcommand{\norm}[1]{\left\| #1 \right\|}

\newcolumntype{d}[1]{D{.}{.}{#1}}

\newtheoremstyle{general}
{3mm} 
{3mm} 
{\it} 
{} 
{\bfseries} 
{.} 
{.5em} 
{} 

\theoremstyle{general}

\newtheorem{lemma}{Lemma}
\newtheorem{theorem}{Theorem}
\newtheorem{corollary}{Corollary}
\newtheorem{assumption}{Assumption}
\newtheorem{remark}{Remark}

\makeatletter
\renewenvironment{proof}[1][\proofname]{\par
    \pushQED{\qed}%
    \normalfont \topsep6\p@\@plus6\p@\relax
    \trivlist
    \item\relax{
        \bfseries
        #1\@addpunct{.}}\hspace\labelsep\ignorespaces
    }{%
     \popQED\endtrivlist\@endpefalse
     }
\makeatother

\title{Coverage of Credible Sets for Regression under Variable Selection}
\author{Samhita Pal, Subhashis Ghosal}\author{Samhita Pal\thanks{
    spal4@ncsu.edu}\hspace{.2cm}\\
    Department of Statistics, North Carolina State University\\
    and \\
    Subhashis Ghosal \\
    Department of Statistics, North Carolina State University}
\date{}
\begin{document}

\maketitle

\begin{abstract}
We study the asymptotic frequentist coverage of credible sets based on a novel Bayesian approach for a multiple linear regression model under variable selection. We initially ignore the issue of variable selection, which allows us to put a conjugate normal prior on the coefficient vector. The variable selection step is incorporated directly in the posterior through a sparsity-inducing map and uses the induced prior for making an inference instead of the natural conjugate posterior. The sparsity-inducing map minimizes the sum of the squared $\ell_2$-distance weighted by the data matrix and a suitably scaled $\ell_1$-penalty term. We obtain the limiting coverage of various credible regions and demonstrate that a modified credible interval for a component has the exact asymptotic frequentist coverage if the corresponding predictor is asymptotically uncorrelated with other predictors. Through extensive simulation, we provide a guideline for choosing the penalty parameter as a function of the credibility level appropriate for the corresponding coverage. We also show finite-sample numerical results that support the conclusions from the asymptotic theory. We also provide the \texttt{credInt} package that implements the method in \texttt{R} to obtain the credible intervals along with the posterior samples.
\end{abstract}

\section{Introduction}

The multiple linear regression model is one of the most useful tools for analyzing data, expressing an overall relationship between a response variable and a set of predictors as an affine function. Not all listed predictor variables are often active in the regression, prompting the need to select the relevant predictors for a better interpretable model and more precise estimation and prediction. The issue of variable selection is especially vital for data consisting of many predictors when a meaningful inference is possible only by incorporating such a step in the analysis, but even for a fixed-dimensional setting, a variable selection step is highly desirable. This leads us to the model selection problem, which has been addressed by both Bayesian and non-Bayesian methods. Classical methods of model selection include forward selection and backward selection. The former begins with the null model and keeps adding predictors in the model sequentially as long as the predictive power goes up significantly, while the latter begins with the full model and sequentially removes unneeded until a stopping criterion is met. Another commonly employed approach involves penalization, which adds a penalty function to the objective function. This encourages the minimizer to produce solutions with greater sparsity than a minimizer without a penalty. The most notable such method is the  
LASSO \citep{tibshirani1996regression}, which penalizes the $\ell_1$-norm of the coefficient vectors to obtain a sparse estimator. Asymptotic properties such as consistency and the limiting distribution of the LASSO were studied in \cite{fu2000asymptotics} in the fixed dimensional setting. If the dimension of the model is large, possibly larger than the sample size, additional compatibility conditions are needed on the predictors and the true value of the vector of regression coefficients to ensure identifiability. For a detailed account of convergence and selection properties of LASSO-type estimators, see \cite{buhlmann2011statistics}. 

In the Bayesian domain, the emphasis has been on formulating a prior distribution that encourages sparsity in the posterior model. The approach of a spike-and-slab prior  \citep{mitchell1988bayesian, ishwaran2005spike} combines a point mass or an approximation to it (spike) at zero and a thick-tailed continuous distribution (slab). Gibbs sampling and a data augmentation technique are used to compute the posterior distribution; see \cite{george1993variable}. A typically faster alternative is to adopt a continuous-shrinkage prior such as the horseshoe prior \citep{carvalho2010horseshoe}. It incorporates a single density function to emulate the concentration at zero and the thickness of the tail in a spike-and-slab prior, simplifying the model by using a one-component approach. Asymptotic concentration and selection properties of the posterior based on a spike-and-slab prior were established in \cite{castillo2015bayesian} and analogous properties for a continuous shrinkage prior in \cite{song2017nearly}. Variational inference for the variable selection problem was recently considered by some authors as an approximate Bayesian inference method providing significantly faster computation; see \cite{ray2022variational, ormerod2017variational,huang2016variational,mukherjee2022variational,yang2020alpha} Concentration properties for estimation and selection consistency properties were studied by \cite{zhang2020convergence, han2019statistical,bai2020nearly}.  \cite{yang2020variational} showed asymptotic frequentist coverage of a credible ball assuming an orthogonal design matrix. 

While penalization methods for variable selection like the LASSO are hugely popular and possess desirable convergence and model selection properties, they do not naturally quantify uncertainty. Even in the fixed dimensional case, because of the bias and degeneracy of the LASSO, the limiting distribution obtained in \cite{fu2000asymptotics} cannot be used to construct a confidence region as the limiting distribution is dependent on the sparsity structure of the true regression coefficients. 
By applying a debiasing technique on the LASSO,  \citep{zhang2014confidence, van2014asymptotically, javanmard2018debiasing} constructed confidence bands, but at the expense of losing sparsity. Modifying the idea of the bootstrap LASSO in \cite{fu2000asymptotics}, \cite{chatterjee2011bootstrapping} showed that the resulting confidence region has good coverage in the fixed-dimensional setting. Confidence regions for the regression vector in the high-dimensional linear regression setting under sparsity appear to have been studied only in \cite{nickl2013confidence} and \cite{cai2017confidence}. Bayesian methods automatically provide a measurement of uncertainty through the posterior distribution. However, the Bayesian way of quantifying the uncertainty may not match a frequentist assessment because these two use very different notions of randomness. For smooth fixed-dimensional parametric families, the Bernstein-von Mises theorem ensures the approximate agreement between Bayesian and frequentist measures of uncertainty, but this may not carry over to non-parametric and other settings; see \cite{cox1993analysis}. In smoothing regimes, asymptotic coverage may be assured only by modifications such as undersmoothing or inflating a credible ball; see  
\cite{knapik2011bayesian}, \cite{castillo2013nonparametric}, \cite{szabo2015frequentist}, \cite{yoo2016supremum}, \cite{ray2017adaptive}, \cite{sniekers2015adaptive}, and others. 
For a high-dimensional setting, the problem of frequentist validation of Bayesian uncertainty quantification was addressed only by \cite{van2017uncertainty}, \cite{belitser2019general},  \cite{belitser2020empirical, changepoint}. 

A novel and convenient Bayesian approach to inference on a restricted subspace through the induced distribution by a restriction-complying ``immersion'' map, which can often be a projection, from a simple unrestricted posterior typically obtained through conjugacy.  \cite{lin2014bayesian} and \cite{chakraborty2021convergence, chakraborty2021coverage, chakraborty2022rates} used the technique for inference on monotone functions;  \cite{bhaumik2015bayesian, bhaumik2017efficient, bhaumik2022two} for differential equation models, and \cite{wang2022coverage, wang2023posterior, wang2023bayesian} for problems on multivariate monotonicity. 
 The corresponding ``projection-posterior'', or a more general ``immersion-posterior'', can be easily computed by conjugate posterior sampling, typically followed by a simple optimization step. For differential equation models, a Bernstein-von Mises theorem holds, justifying the uncertainty quantification of the projection-posterior from the frequentist principle. For shape-restricted models, the asymptotic coverage is slightly higher than the credibility level, the opposite of the phenomenon observed in  \cite{cox1993analysis} for smoothing problems. Interestingly, the relationship between the limiting coverage and credibility does not depend on the true function. Hence, a target exact asymptotic coverage can be obtained by starting from a predetermined lower level (\cite{chakraborty2021coverage,chakraborty2022rates}, \cite{wang2022coverage,wang2023bayesian}). 
 In a linear regression model with variable selection, the need to address sparsity in the prior makes posterior distributions harder to compute and analyze theoretically. This is particularly apparent because results on coverage of Bayesian credible regions are largely absent from the literature. The problem can be alleviated by using an immersion posterior to make an inference. In the absence of sparsity, a multivariate normal prior is conjugate for the linear regression problem, leading to an explicit expression for the posterior distribution. An immersion map given by the minimizer of the sum of a weighted squared Euclidean distance and an $\ell_1$-penalty induces sparsity in the solution and is appropriate for the problem.

This paper studies the construction and frequentist coverage properties of credible sets for the coefficient vector using the sparse projection-posterior method under the fixed dimension setting. We derive a weak limit of the immersion-posterior distribution with an interesting structural similarity with the distributional limit of the LASSO derived in \cite{fu2000asymptotics}. Using the weak limit, we derive an expression for the limiting coverage of a Bayesian credible ball. When a predictor variable is asymptotically uncorrelated with the remaining predictors, we simplify the expression for the coverage and compare it with the credibility level. While the limiting coverage is lower than the credibility level if the coefficient is non-zero, we can identify a higher credibility level, depending on the tuning parameter, such that the limiting coverage coincides with the targeted coverage. The corresponding limiting coverage is higher if the parameter assumes zero value, provided that the intended coverage level is sufficiently high. Apart from reconciling the Bayesian and frequentist notions of uncertainty quantification as in the Bernsterin-von Mises theorem for smooth fixed-dimensional parametric families, this result provides an automatic tuning mechanism. 

The rest of the paper is organized as follows. The setup and the methodology are formally described in the next section. In \Cref{main results}, we derive the limiting results for establishing asymptotic coverage of the credible balls, stating the necessary assumptions and confirming their pragmatism. \Cref{simulation} shows finite-sample results of the proven asymptotic guarantees and provides a table for guidelines for choosing the correct penalty parameter. Finally, in \Cref{conclusion}, we discuss the contribution of this work and the future directions, followed by the proofs in the appendix.

\section{Setup and Methodology}
\label{methodology}

We consider the problem of model selection and uncertainty quantification of the chosen variables in the context of the linear regression model 
\begin{align}
\label{linear model}
    \bY = \bX\btheta + \bm{\varepsilon}, \quad \bm{\varepsilon} \sim \normal_n(\bm{0},\sigma^2\bm{I}_n),
\end{align} 
where $\bY \in \R^n$ is the response vector observed for $n$ samples, $\bX \in \R^{n \times p}$ is a deterministic design matrix and $\btheta \in \R^p$ is the vector of regression coefficients. Regarding the methodology, the deterministic or random nature of the predictor variables is immaterial, but it matters in the study of asymptotic properties. If the predictors are random, the conditions imposed on the predictor variables will have to be satisfied with probability tending to one. 

Throughout the paper, the number of explanatory variables $p$ in this model is assumed to be fixed and smaller than the sample size $n$. A smaller number $s$ of these $p$ predictors are active, but $s$ and the active predictor indices are unknown. Although the regression vector may be estimated at the parametric rate $n^{-1/2}$, identifying the model consisting only of the active predictors is desirable since it gives a more interpretable model. Let $S_0$ consist of the indices corresponding to active predictor variables in the model. 
We assume that the data $\bY=(Y_1,\ldots,Y_n)^{\mathrm{T}}$ is generated through the process 
\begin{align}
\label{true distribution}
    Y_i = \bX_i^{\mathrm{T}}\btheta^0 + \varepsilon_i, \quad \varepsilon_1,\ldots,\varepsilon_n  \mbox{ are i.i.d. with mean 0 and variance } \sigma_0^2,
\end{align} 
where $\bX_i^{\mathrm{T}}$ stands for the $i$th row of $\bX$ (i.e., the vector of covariates for the $i$th observation). Thus, the true data generating process follows the linear 
model \eqref{linear model} with the true value of the regression coefficients $\btheta^0$ and the true value of the error standard deviation $\sigma_0$, except that the error distribution need not be normal. 

Throughout the paper, we will use the following notations. Probability under the true model will be denoted by $\mathbb{P}_0$ and the corresponding expectation by $\mathbb{E}_0$. 
The normalized cross-products matrix is given by $\bC_n = n^{-1}{\bm{X}^{\mathrm{T}}\bm{X}}$. For notational convenience in theoretical analysis, we shall place the $s_0$ active predictors corresponding to non-zero regression coefficients at the beginning of the arrangement: let $\bX_{(1)}$ stand for the submatrix of the data matrix $\bX$ formed by the active predictors and $\bX_{(2)}$ stand for the submatrix corresponding to the remaining $p-s_0$ variables with zero true coefficients. Then, the matrix $\bC_n$ can be split into block matrices  
\begin{align}
\label{sample covariance}
 \begin{bmatrix}
\bC_{n(11)} & \bC_{n(12)} \\
\bC_{n(21)} & \bC_{n(22)} 
\end{bmatrix} = \frac{1}{n} \begin{bmatrix}
\bX_{(1)}^{\mathrm{T}} \bX_{(1)} & \bX_{(1)}^{\mathrm{T}} \bX_{(2)} \\
\bX_{(2)}^{\mathrm{T}} \bX_{(1)} & \bX_{(2)}^{\mathrm{T}} \bX_{(2)} 
\end{bmatrix}. 
\end{align}
 The $q$-norm of a vector $\bm{a} \in \R^d$ is given by
 $\norm{\bm{a}}_q = (\sum_{i = 1}^{d} |a_i|^q)^{1/q}$. For $q=2$, this corresponds to the Euclidean norm, which is denoted by $\|\cdot\|$. For a matrix $\bm{M} \in \R^{d_1\times d_2}$, we denote by  $\bm{M}_A$ the submatrix of containing the columns corresponding to a given set $A \subset \{1,2,\dots,d_2\}$. We denote the zero vector by $\bm{0}$ and the identity matrix of order $r$ by $\bm{I}_r$. We use the $\rightsquigarrow$ notation to denote weak convergence. For a sequence $a_n$ and positive sequence $b_n$, let $a_n =\mathcal{O}(b_n)$ mean that $|a_n|\le C b_n$ for some constant $C>0$. For a sample space $\mathfrak{X}$, let $\mathfrak{M}(\mathfrak{X})$ stand for the space of all probability measures on $\mathfrak{X}$. The distributional law of a random element $T$ will be denoted by $\mathcal{L}(T)$. The symbol $\mathrm{P}$ refers to a probability statement with respect to a generic probability distribution. Let $\delta_0$ denote the Dirac delta measure at the point $0$. We shall write $\Phi$ for the standard normal cumulative distribution function and $z_\alpha=\Phi^{-1}(1-\alpha)$ for the $(1-\alpha)$-quantile of it. 
 
The proposed immersion-posterior methodology for the variable selection problem differs from the classical Bayesian approach of putting a prior taking sparsity into account and then updating to the posterior distribution that puts weights on different models corresponding to different sets of selected predictors. In the proposed approach, we put a conditionally conjugate normal prior on the parameter vector given, disregarding the issue of variable selection at first, and then a conjugate inverse gamma prior for $\sigma^2$: 
Further, we put an inverse gamma prior for $\sigma^2$, that is, 
\begin{align}
\label{prior}
\btheta|\sigma \sim \normal_p(\textbf{0}, \sigma^2 a_n^{-1} \bm{I}_p), \qquad  
    \sigma^{-2}\sim \mathrm{Ga}(b_1,b_2)
\end{align}
for some sequence $a_n>0$ and constants $b_1,b_2>0$. The non-informative choice $b_1=b_2=0$, which corresponds to the density $\sigma^{-1}$ for $\sigma$, may also be used. The choices of $b_1$ and $b_2$ will not matter for the intended asymptotic study in this paper. 

This results in the ``unrestricted'' posterior distribution given by 
\begin{align}
\label{pos dist}
    \btheta|(\bY,\sigma)\sim  \normal_{p} \big(\big(\bX^{\mathrm{T}} \bX + a_n\bm{I}_{p}\big)^{-1} \bX^{\mathrm{T}} \bY, \sigma^2 \big(\bX^{\mathrm{T}} \bX + a_n\bm{I}_{p}\big)^{-1}\big).
\end{align} 
Further, the marginal posterior for $\sigma$ is given by 
\begin{align}
\label{posterior sigma}
    \sigma^{-2}|\bY\sim  \mathrm{Ga} \big( b_1+\frac{n}{2},  b_2+ \frac{1}{2} \bY^{\mathrm{T}} (\bm{I}_n-\bX(\bX^{\mathrm{T}} \bX+a_n \bm{I}_p)^{-1} \bX^{\mathrm{T}})  \bY   \big).
\end{align} 

We denote the posterior mean $\big(\bX^{\mathrm{T}} \bX + a_n\bm{I}_{p}\big)^{-1} \bX^{\mathrm{T}} \bY$ by $\hat{\btheta}^\mathrm{R}$ owing to its interpretation as the ridge-regression estimator. The posterior distribution will be denoted by $\Pi(\cdot|\bY)$. 

Note that the posterior distribution disregards variable selection because the prior does not introduce such a mechanism. The issue is then addressed by ``correcting'' the unrestricted posterior through an immersion map that transforms dense vectors into sparse vectors, in that the distribution induced by this map from the unrestricted posterior is used for inference. We choose the sparsity-inducing immersion map given by 
\begin{align}
\label{sparse_proj}
\iota:\btheta \mapsto \btheta^* \coloneqq \argmin_{\bm{u}} \{ n^{-1} \norm{\bX \btheta - \bX \bm{u}}^2 + \lambda_n \norm{\bm{u}}_1\}.
\end{align}
We shall call the map $\iota$ the sparse-projection operator. It depends on the choice of the tuning parameter $\lambda_n$.   
The map is motivated by the LASSO in that if $\btheta=\hat\btheta$, the least square estimator, then $\btheta^*$ is the LASSO estimator $\hat\btheta^{\mathrm{L}}$ defined by 
\begin{align}
\hat\btheta^{\mathrm{L}} &= \argmin_{\bm{u}} \{ n^{-1} \norm{\bY - \bX \bm{u}}^2 + \lambda_n \norm{\bm{u}}_1\}\nonumber \\
& =\argmin_{\bm{u}} \{ n^{-1} \|\bX \hat{\btheta} - \bX \bm{u}\|^2 + \lambda_n \norm{\bm{u}}_1\}.
\label{lasso}
\end{align}
The distribution induced from the unrestricted posterior by the map $\iota$ is an immersion posterior distribution in the terminology of \cite{wang2022coverage}. In our context, it will be called the sparse-projection posterior distribution and will be denoted by $\Pi^*(\cdot|\bY)$, that is, $\Pi^*(B|\bY)=\Pi(\btheta^*\in B|\bY)$ for $B\subset \mathbb{R}^p$. 

The immersion posterior is easy to compute by sampling.  
Having obtained a dense posterior sample $\btheta$ from \eqref{pos dist}, we compute $\btheta^*=\iota(\btheta)$ and record that as a sample from the immersion posterior distribution. We repeat the operation independently sufficiently many times so that posterior probabilities and expectations can be reliably calculated from sampling. In particular, the method also computes model posterior probabilities. 
As the method does not use any Markov chain Monte Carlo (MCMC) method, it promises faster computation. The method is also amenable to parallelization since both $\bX^{\mathrm{T}} \bX$ and $\bX^{\mathrm{T}} \bY$ can be computed by dividing the dataset into several parts and computing the sum of products at different machines, followed by an aggregation step at the central server. This feature is especially useful if the sample size is huge. 

It may be noted that the immersion posterior method may also be used in conjugation with other sparsity-inducing operators instead of \eqref{sparse_proj}. Other penalties that may be used to define an immersion posterior include those appearing in defining estimators alternative to the LASSO such as the Minimax Concave Penalty (MCP) \citep{zhang2010nearly}, the Smoothly Clipped Absolute Deviation (SCAD) \citep{fan2001variable}, the Dantzig selector \citep{candes2007dantzig}, the adaptive LASSO \citep{zou2006adaptive}, the non-negative garrotte estimator \citep{breiman1995better, yuan2007non} among others. In the present paper, we forgo 
other possible penalty functions and 
study the asymptotic coverage posterior credible regions corresponding only to the immersion map \eqref{sparse_proj}.

\section{Main Results}
\label{main results}

\setcounter{equation}{0} 

Let the predictor dimension $p$ be fixed. We make the following assumptions: 

\begin{assumption}[Predictor]
\label{predictor}
The matrix $\bC_n\to \bC$ for some positive definite matrix $\bC$ and $n^{-1}\max\{\|\bX_i\|^2:i=1,\ldots,n\}\to 0$.
\end{assumption}

\begin{assumption}[Tuning]
\label{tuning}
The tuning parameter $\lambda_n$ satisfies $\lambda_n \sqrt{n}\to \lambda_0$ for some $\lambda_0\ge 0$.
\end{assumption}

Under the above conditions, \cite{fu2000asymptotics} showed that the LASSO estimator $\hat{\boldsymbol\theta}^\mathrm{L}$ is weakly consistent and the limiting distribution of $\sqrt{n}(\hat\btheta^{\mathrm{L}}-\btheta^0)$ puts a positive probability at zero for the components corresponding to the irrelevant variables. 

\begin{theorem}[\cite{fu2000asymptotics}]
Let $p$ be fixed, and Assumptions~\ref{predictor} and \ref{tuning} hold. Then  $\hat{\bm\xi}_n:=\sqrt{n}(\hat{\btheta}^\mathrm{L} - \btheta^0)\rightsquigarrow \bm{\xi}$, where  
\begin{align}
\label{lasso_lim_dist}
    \bm\xi = \argmin_{\bm{v}\in \R^p} \{\bm{v}^\mathrm{T}\bC\bm{v} - 2\sigma_0\bm{v}^\mathrm{T}\bC^{1/2}\bm{\Delta} + \lambda_0  \big[ \sum_{j=1}^{s_0} v_{j}\mathrm{sign}(\theta^0_{j}) + \sum_{j=s_0+1}^{p} |v_j| \big]\},
\end{align} 
with $\bm\Delta \sim \normal_p(\bm{0},\bm{I}_p)$.
\end{theorem}

It also follows that under the stated conditions, the variance estimator $\hat{\sigma}^2=n^{-1} \|\bY-\bX \hat\btheta\|^2$ is $\sqrt{n}$-consistent for $\sigma^2$, where $\hat\btheta$ is the least square estimator or a ridge-regression estimator $\hat\btheta^{\mathrm{R}}$; see Remark~\ref{root-n-consistency}. Assumption~\ref{predictor} holds with probability tending to one for random predictors if the values are sampled independently from a fixed nonsingular distribution. The second part of the assumption clearly holds if the predictors are uniformly bounded, a restriction often imposed by scaling.

Under the same condition on the tuning parameter ${\lambda_n}{\sqrt{n}} \to \lambda_0$ for some constant  $\lambda_0 \geq 0$, and assuming that $a_n/\sqrt{n}\to 0$, we obtain below a joint weak limit of the posterior distribution of $\sqrt{n}(\btheta^*-\btheta^0)$ given the data in the space $\mathfrak{M}(\mathbb{R}^p)$, equipped with the topology of weak convergence, and the distribution of the normalized LASSO estimator in $\mathbb{R}^p$. 

\begin{theorem}
\label{limit distribution}
Let Conditions \ref{predictor} and \ref{tuning} hold, and  $a_n/\sqrt{n}\to 0$. Then 
\begin{align}
 \label{joint distr}
    \big(\Pi(\sqrt{n}(\btheta^*-\btheta^0)\in \cdot|\bY), \sqrt{n}(\hat{\btheta}^{\mathrm{L}}-\btheta^0)\big) \rightsquigarrow \big(\mathcal{L}(\bm{T}^* \in \cdot |\bm{\Delta}),\bm\xi\big) 
\end{align} 
on the product space $\mathfrak{M}(\R^p) \times \R^p$, 
where 
\begin{align}
\label{projection_lim_rv_small_p}
    \bm{T}^* = \argmin_{\bm{t}\in \R^p} \{ \bm{t}^{\mathrm{T}} \bC \bm{t} - 2 \bm{t}^{\mathrm{T}} \bm{C} \bm{W}^* + \lambda_0 \big[ \sum_{j=1}^{s_0} t_{j}\textnormal{sign}(\theta^0_{j}) + \sum_{j=s_0+1}^{p} |t_j| \big] \},
\end{align} 
for $\bm{W}^*|\bm{\Delta}\sim \normal(\sigma_0 \bm{C}^{-1/2}\bm{\Delta}, \sigma_0^2 \bm{C}^{-1})$ and $\bm\xi$ is defined in \eqref{lasso_lim_dist}. 
\end{theorem}

As shifting the random distribution by $\hat{\bm\xi}_n=\sqrt{n}(\hat\btheta_n-\btheta^0)$ is a continuous operation in $\mathfrak{M}(\mathbb{R}^p)$, we conclude the following. 

\begin{corollary}
\label{centered posterior}
Under the conditions of \Cref{limit distribution}, 
\begin{align}
    \label{centered joint distr}
    \big(\Pi(\sqrt{n}(\btheta^*-\hat\btheta^\mathrm{L})\in \cdot|\bY), \sqrt{n}(\hat{\btheta}^{\mathrm{L}}-\btheta^0)\big) \rightsquigarrow \big(\mathcal{L}(\bm{T}^* -\bm{\xi}\in \cdot |\bm{\Delta}),\bm\xi\big). 
    \end{align}
\end{corollary}

We show in \Cref{pos_prob_at_0} that when $\theta^0_{j} = 0$ for all $j = s_0+1,\dots,p$, the minimizer $\bm{T}^*$ assumes the value $0$ for its last $p-s_0$ components with positive probability. The distribution has no closed-form expression except for some special cases, and the statement is only marginally useful. However, we can compute the limiting frequentist coverage of a Bayesian credible region defined by the quantiles of the posterior distribution of $\sqrt{n}(\btheta^*-\btheta^0)$, which can be obtained from \Cref{limit distribution}. 

For $\bar\btheta\in \mathbb{R}^p$ and $r>0$, let  
\begin{align}
\label{credible set}
    B_\mathcal{K}(\bar{\btheta}, r) = \{\btheta: \|\sqrt{n}(\btheta - \bar{\btheta}) \|_{\mathcal{K}} \le r\},  
\end{align} 
where the seminorm $\|\cdot\|_{\mathcal{K}} $ is defined by 
\begin{align} 
\label{Minkowski}
\|\bm{x}\|_{\mathcal{K}}=\inf\{\lambda > 0 : \bm{x}/{\lambda} \in \mathcal{K}\},  
\end{align}
the Minkowski functional of a convex set $\mathcal{K}$ containing $0$ in its interior. 
Define a $(1-\alpha)$-credible set to be $B_\mathcal{K}(\hat{\btheta}^{\mathrm{L}}, r_{1-\alpha})$, where $r_{1-\alpha}$ is the $(1-\alpha)$-quantile of the posterior distribution of $\|\sqrt{n}(\btheta^* - \hat{\btheta}^{\mathrm{L}}) \|_{\mathcal{K}}$, that is 
\begin{align} 
\label{quantile of radius}
\Pi^*(B_\mathcal{K}(\hat{\btheta}^{\mathrm{L}}, r_{1-\alpha})|\bY)=
\Pi(\btheta^*: \|\sqrt{n}(\btheta^* - \hat{\btheta}^\mathrm{L}) \|_{\mathcal{K}} \le  r_{1-\alpha}|\bY) = 1 - \alpha. 
\end{align}
If the equality cannot be obtained, we choose $r_{1-\alpha}$ the infimum of all $r>0$ such that $\Pi^*(B_\mathcal{K}(\hat{\btheta}^{\mathrm{L}}, r)|\bY)>1-\alpha$. 
For different choices of $\mathcal{K}$, we can have coverage results for different norms or pseudo-norms. For example, choosing $\mathcal{K} = [-1,1]^p$, we get the max norm $\|\cdot\|_{\infty}$, choosing $\mathcal{K}$ as a unit sphere, we get the unit norm $\|\cdot\|_2$, choosing $\mathcal{K}$ as a unit diamond, we get the $\ell_1$-norm $\|\cdot\|_1$, and choosing $\mathcal{K} = \R\times \cdots\times \R \times [-1,1] \times \R \times \cdots\times \R$, we get credible intervals for individual components. 

\begin{theorem}
\label{coverage}
    For a given level of confidence $(1-\alpha)$ and $\bm\Delta \sim \normal_p(\bm{0},\boldsymbol{I}_p)$, we have under Assumptions \ref{predictor} and \ref{tuning} and $a_n/\sqrt{n}\to 0$,
    \begin{eqnarray}
    \lefteqn{\mathrm{P}\left(\mathrm{P}(\|\bm{T}^* - \bm\xi\|_\mathcal{K} \leq \|\bm\xi\|_\mathcal{K} \big| \bm\Delta) < 1 - \alpha\right)}\nonumber \\
    && \le \liminf_{n\to\infty} \mathbb{P}_{0}\big(\btheta^0 \in B_\mathcal{K}(\hat{\btheta}^\mathrm{L}, r_{1-\alpha})\big) 
     \nonumber \\
    && \le \limsup_{n\to\infty} \mathbb{P}_{0}\big(\btheta^0 \in B_\mathcal{K}(\hat{\btheta}^\mathrm{L}, r_{1-\alpha})\big) \nonumber \\
    && \le \mathrm{P}\left(\mathrm{P}(\|\bm{T}^* - \bm\xi\|_\mathcal{K} \leq \|\bm\xi\|_\mathcal{K} \big| \bm\Delta) \le  1 - \alpha\right),
     \label{eq:coverage} 
    \end{eqnarray}
    where $\bm{T}^*$ and $\bm\xi$ are as in Theorem~\ref{limit distribution}. 
\end{theorem}

The limiting coverage of a sparse-projection posterior credible region is thus assured to be at least the probability on the right side, but it may not equal the credibility level $1-\alpha$. The limiting coverage depends on $\bC$, $\sigma_0$, $\lambda_0$ and $\alpha$. In the following subsection, we evaluate the limiting coverage explicitly for a credible interval for a component under an asymptotic orthogonality setting. 

\subsection{Special Case: Asymptotically Uncorrelated Predictor}

We now explicitly evaluate the limiting coverage of a credible interval of a single regression coefficient $\theta_j$, assuming that the $j$th predictor is asymptotically uncorrelated with the rest, for some $j=1,\ldots,p$.  

\begin{assumption}[Asymptotically uncorrelated predictor]
\label{asymp_indep}
    The $j$th predictor variable $X_j$ is asymptotically uncorrelated with the remaining predictors $\{X_1,\ldots,X_{j-1},X_{j+1},\ldots,X_p\}$, that is,    $\bC=\lim_{n\to\infty} n^{-1}\bX^\mathrm{T}\bX$ is block-diagonal with blocks $\{j\}$ and $[-j]:=\{1,\ldots,j-1,j+1,\ldots,p\}$.
\end{assumption}

Let $c_j$ stand for the $j$th diagonal entry of $\bC$. Note that, under the above assumption, all other entries of the $j$th row and $j$th column of $\bC$ are 0.

Although the assumption of asymptotic uncorrelatedness is very special, it may naturally hold for predictors not related to other predictors. Moreover, the condition can be enforced by an orthogonalization technique, but it then changes the meaning of the coefficient.  

Consider the credible interval $ [\hat\theta^{\mathrm{L}}_j-r_{1-\alpha,j},\hat\theta^{\mathrm{L}}_j+r_{1-\alpha,j}]$, where $\hat\theta^{\mathrm{L}}_j$ is the $j$th coordinate of the LASSO estimator $\hat\btheta^{\mathrm{L}}$ and $r_{1-\alpha,j}$ is the $(1-\alpha)$-quantile of the induced posterior distribution of $|\theta^*_j-\hat\theta_j^{\mathrm{L}}|$. This credible interval corresponds to the credible region in \eqref{credible set} with $\mathcal{K}=\R\times \cdots\times \R \times [-1,1] \times \R \times \cdots\times \R$, where $[-1,1]$ appears at the $j$th coordinate. Under Assumption~\ref{asymp_indep}, it follows from 
\eqref{lasso_lim_dist} that 
the $j$th component $\xi_j$ of $\bm\xi$ is 
given by 
\begin{align}
  \xi_j & =\mathrm{argmin} \{ c_j v_j^2 -2\sigma_0 \sqrt{c_j} \Delta_j v_j+\lambda_0 [v_j \mathrm{sign}(\theta_j^0)+|v_j|\mathbbm{1}(\theta_j=0)]: v_j\in \R\}  \nonumber\\
  &=\begin{cases} 
  \frac{\sigma_0}{\sqrt{c_j}}\Delta_j-\frac{\lambda_0}{2 c_j}, & \mbox{ if }\theta_j^0>0,\\
  \frac{\sigma_0}{\sqrt{c_j}}\Delta_j+\frac{\lambda_0}{2 c_j}, & \mbox{ if }\theta_j^0<0,\\
  (\frac{\sigma_0}{\sqrt{c_j}}\Delta_j-\frac{\lambda_0}{2 c_j})\mathbbm{1}\{\Delta_j>\frac{\lambda_0}{2 \sigma_0\sqrt{c_j}}\} \\ 
  + (\frac{\sigma_0}{\sqrt{c_j}}\Delta_j+\frac{\lambda_0}{2 c_j})\mathbbm{1}\{\Delta_j<-\frac{\lambda_0}{2 \sigma_0\sqrt{c_j}}\}, & \mbox{ if }\theta_j^0=0.
  \end{cases}
  \label{jth_lasso_lim_dist}
\end{align}
Thus, $\xi_j$ depends on $\bm{\Delta}$ only through its $j$ component $\Delta_j$ and is independent of the remaining components. 

Similarly, from \eqref{projection_lim_rv_small_p}, it follows that 
the $j$th coordinate $T_j^*$ of $\bm{T}^*$ is given by 
\begin{align}
  T_j^* & =\mathrm{argmin} \{ c_j t_j^2 -2 c_j W_j^* t_j+\lambda_0 [t_j \mathrm{sign}(\theta_j^0)+|t_j|\mathbbm{1}(\theta_j=0)]: t_j\in \R\}  \nonumber\\
  &=\begin{cases} 
  W_j^*-\frac{\lambda_0}{2 c_j}, & \mbox{ if }\theta_j^0>0,\\
  W_j^*+\frac{\lambda_0}{2 c_j}, & \mbox{ if }\theta_j^0<0,\\
  (W_j^*-\frac{\lambda_0}{2 c_j})\mathbbm{1}\{W_j^*>\frac{\lambda_0}{2 c_j}\} \\ 
  + (W_j^*+\frac{\lambda_0}{2 c_j})\mathbbm{1}\{W_j^*<-\frac{\lambda_0}{2 c_j}\}, & \mbox{ if }\theta_j^0=0,
  \end{cases}
  \label{jth_bayes_lim_dist}
\end{align}
where $W_j^*$ is the $j$th component of $\bm{W}^*$, and that $W_j^*|\Delta_j \sim \normal (\Delta_j,1)$. 

We then introduce the notations 
\begin{align} 
h_+(\lambda_0,\zeta) &=2\Phi(|\zeta-\lambda_0/2|)-1, \nonumber \\ 
h_-(\lambda_0,\zeta) &=2\Phi(|\zeta+\lambda_0/2|)-1=h_+(\lambda_0,-\zeta),\nonumber\\
\label{h0}
h_0(\lambda_0,\zeta)&=\begin{cases} 
  \Phi(\zeta-\lambda_0/2)-\Phi(-\zeta-\lambda_0/2), & \mbox{ if }\zeta>\lambda_0/2,\\
 \Phi(-\zeta+\lambda_0/2)-\Phi(\zeta+\lambda_0/2), & \mbox{ if }\zeta<\lambda_0/2,\\
 \Phi(\zeta+\lambda_0/2)-\Phi(\zeta-\lambda_0/2), & \mbox{ if }|\zeta|\le \lambda_0/2.\\
  \end{cases}
\end{align}
and 
\begin{align} 
    \psi(\alpha,\lambda_0)& =\Phi(\lambda_0/2 + z_{\alpha/2})-\Phi(\lambda_0/2 - z_{\alpha/2}), 
    \label{nonzero limit}\\
    \psi_0(\alpha,\lambda_0)& = \int_{-\infty}^{\infty} \mathbbm{1}\big(h_0(\lambda_0,\zeta) \leq 1 - \alpha\big) \phi(\zeta) d\zeta.
    \label{zero limit}
    \end{align}

The asymptotic coverage of the credible interval $ [\hat\theta^{\mathrm{L}}_j-r_{1-\alpha,j},\hat\theta^{\mathrm{L}}_j+r_{1-\alpha,j}]$ for $\theta_j$ is characterized in the following theorem.

\begin{theorem}
\label{coverage of credible interval} 
If Assumptions~\ref{predictor}, \ref{tuning} and \ref{asymp_indep} hold and $a_n/\sqrt{n}\to 0$, then 
for a given level of confidence $0<1-\alpha<1$,  
    \begin{align}
    \mathbb{P}_{0}\big( \theta_{0j}\in  [\hat\theta_j^{\mathrm{L}}-r_{1-\alpha,j},\hat\theta_j^{\mathrm{L}}+r_{1-\alpha,j}]  \big) \to 
    \begin{cases}  \psi(\alpha, \lambda_0\sqrt{c_j}/\sigma_0), & \mbox{ if }\theta_j\ne 0,\\
      \psi_0(\alpha, \lambda_0\sqrt{c_j}/\sigma_0), & \mbox{ if }\theta_j= 0. 
    \label{explicit coverage}
    \end{cases} 
    \end{align}
\end{theorem}

\begin{remark}\rm 
\label{scaling}
    The original data generating process $$Y_i=\sum_{k=1}^p X_{ik}\theta_k+\varepsilon_i=
    X_{ij}\theta_j+\sum_{k=1: k\ne j}^p X_{ik}\theta_k+\varepsilon_i$$ can be rewritten as 
    \begin{align}
        \tilde{Y}_i=\tilde{X}_{ij} \tilde{\theta}_{0j}+\sum_{k=1: k\ne j}^p X_{ik}\frac{\theta_{0k}}{\sigma_0}+\tilde{\varepsilon}_i,
        \label{transformed linear model}
    \end{align}
    where $\tilde{Y}_i=Y_i/\sigma_0$, $\tilde{X}_{ij}=x_{ij}/\sqrt{c_j}$, $\tilde{\varepsilon}_i=\varepsilon_i/\sigma_0$ and $\tilde{\theta}_{j}^0=\theta_{j}^0\sqrt{c_j}/\sigma_0$. For the transformed process, the $j$th predictor $\tilde{X}_j$ still remains asymptotically uncorrelated with the remaining predictors, and the corresponding $\bC$-matrix has the $j$th diagonal entry 1. The true error variance also reduces to 1. Thus, the asymptotic coverage for the credible interval for $\theta_j$ using the penalty parameter $\lambda_0$ coincides with that for $\tilde{\theta}_j=\theta_j\sigma_0/\sqrt{c_j}$ using the penalty parameter $\lambda_0\sqrt{c_j}/\sigma_0$. Hence, to prove Theorem~\ref{coverage of credible interval}, we can assume, without loss of generality, that $c_j=1$ and $\sigma_0 = 1$. 
\end{remark}  

Below, in \Cref{fig:small_p_coverage}, we provide the plots of $h_+(\lambda_0,\zeta)$ and $h_0(\lambda_0,\zeta)$ against $\zeta$ for different choices of the limiting penalty $\lambda_0$ as well as the limiting coverage probabilities against the confidence level for the signal and noise components. Since $\psi(\alpha,\lambda_0)\le 1-\alpha$ with equality holding only for $\lambda_0=0$,  for the signal components, no non-zero value of the penalty parameter $\lambda_0$ can lead to the exact $(1-\alpha)$ asymptotic coverage, or more, of a $(1-\alpha)$ sparse-projection posterior credible interval. However, by fixing a penalty parameter $\lambda_0$, we can ensure that for some $\gamma$ depending on $\lambda_0$ and $\alpha$, a $(1-\gamma)$ credible interval will have exactly $(1-\alpha)$ asymptotic coverage if $\theta^0_j\ne 0$. Moreover, in \Cref{fig:cov_vs_lam_for_diff_alphas}, the function $\psi(\alpha,\lambda_0)$ is shown to be uniformly bounded above by $\psi_0(\alpha,\lambda_0)$ for a given $\alpha$, meaning that the limiting coverage of a noise component will always be higher than that of a signal component. Thus, owing to the fact that  $\psi_0(\alpha,\lambda_0) \ge  \psi(\alpha,\lambda_0)$, as shown by numerical evaluation in \Cref{fig:cov_vs_lam_for_diff_alphas}, it follows that the asymptotic coverage of the corresponding interval is at least $(1-\alpha)$. 

\begin{figure}[htbp]
    \centering
    \includegraphics[width = 0.4\linewidth, height = 4.5cm]{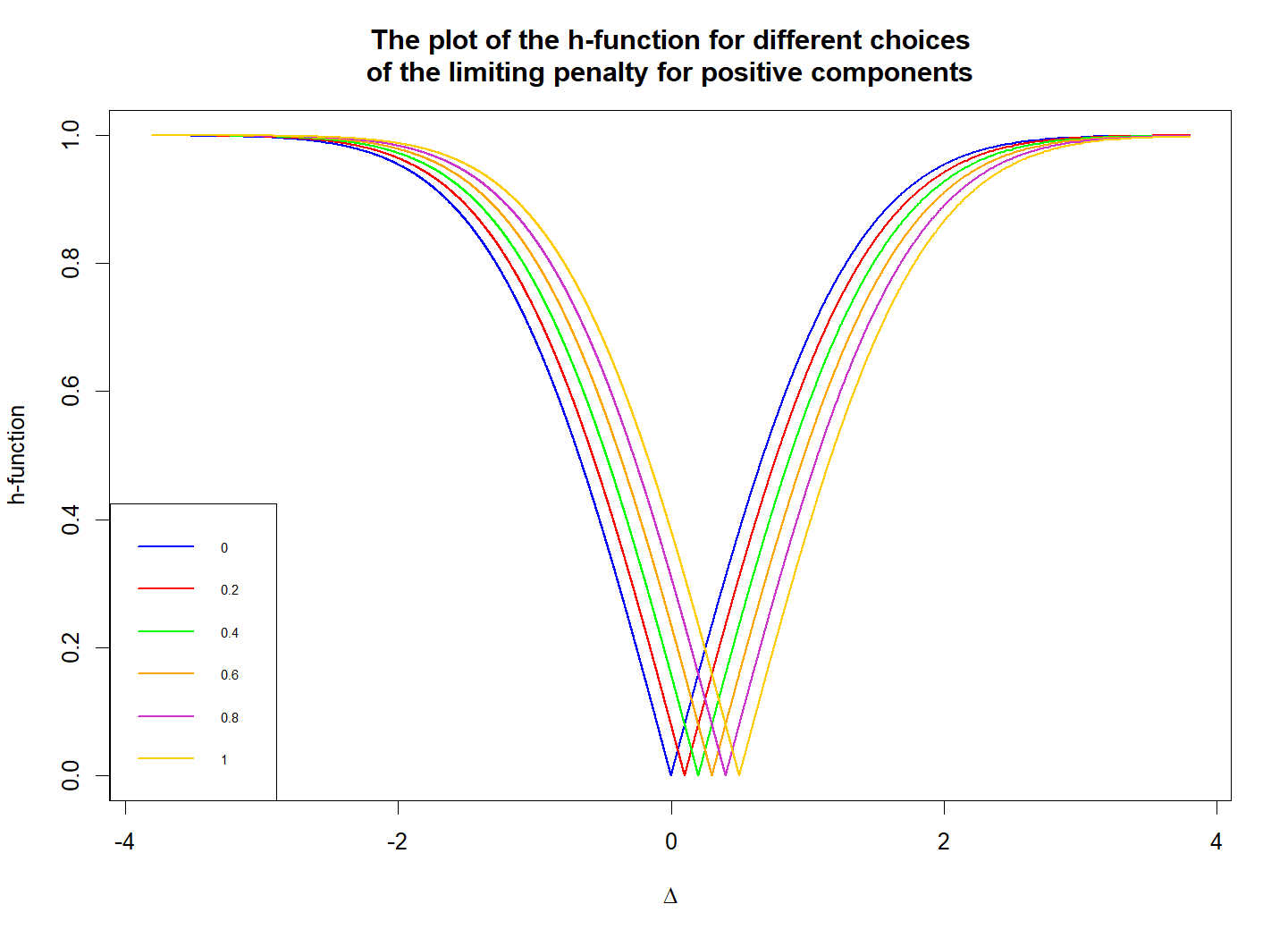}
    \includegraphics[width = 0.4\linewidth, height = 4.5cm]{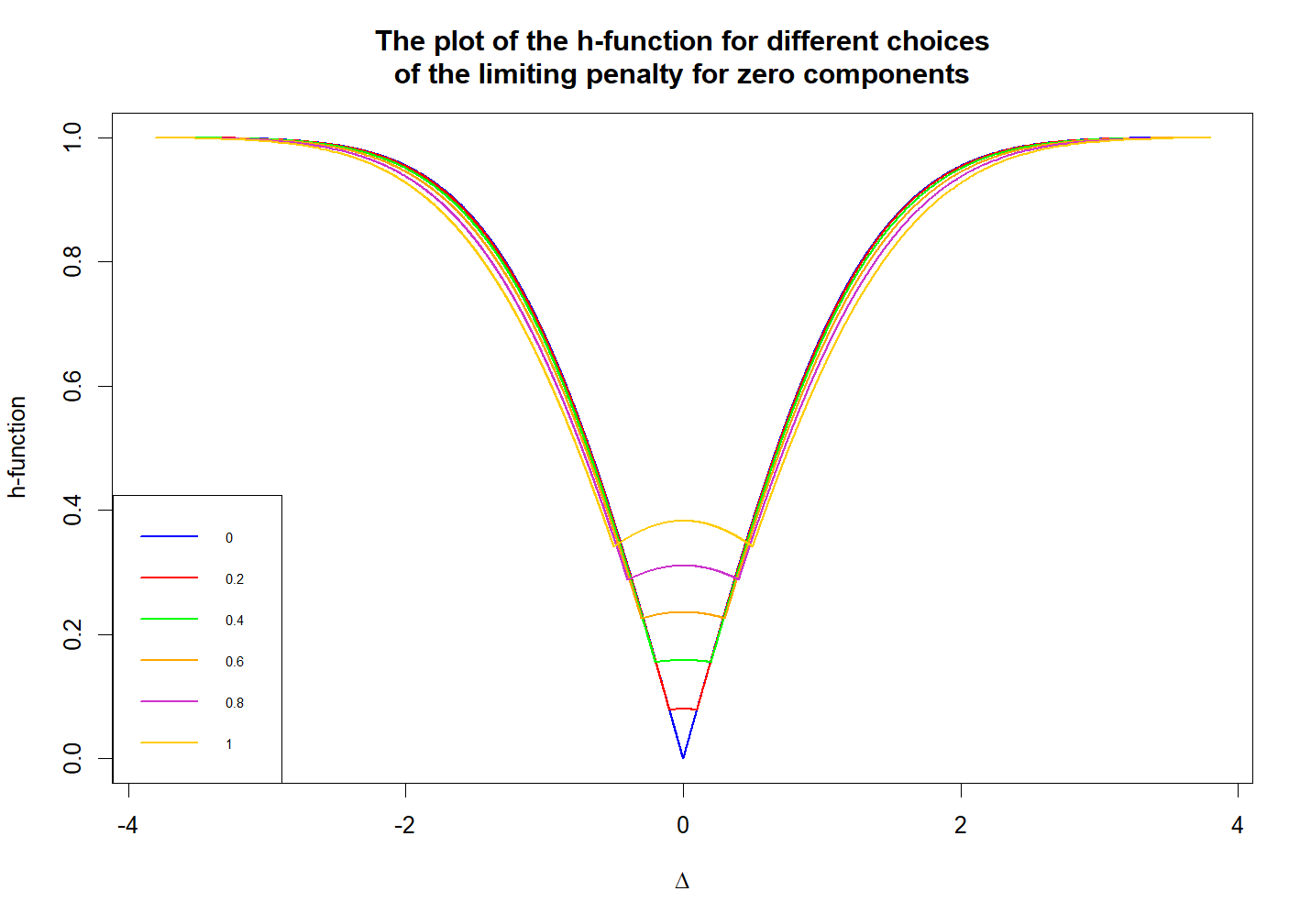}
    \includegraphics[width = 0.4\linewidth, height=0.4\linewidth]{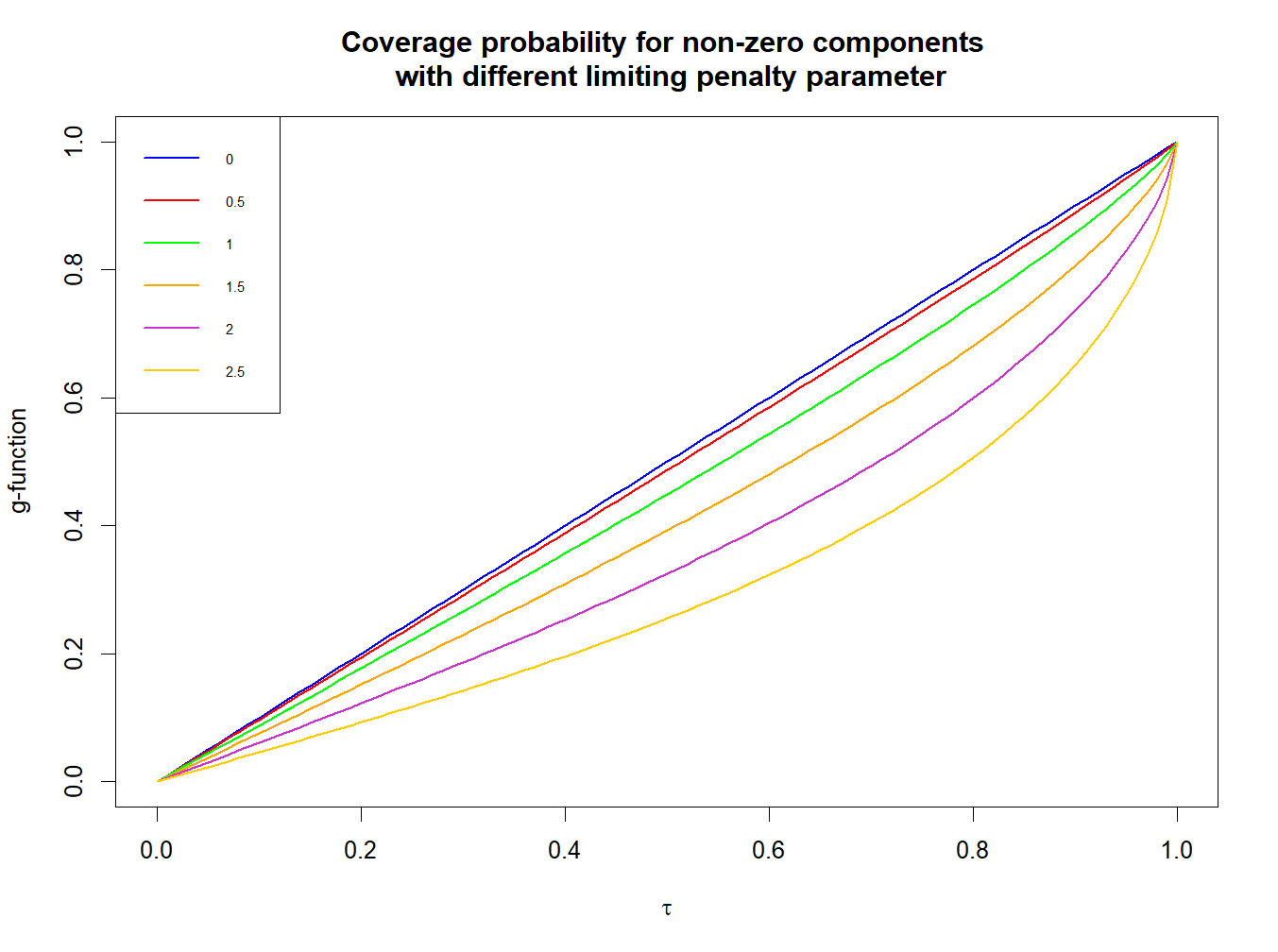}
    \includegraphics[width = 0.4\linewidth, height=0.4\linewidth]{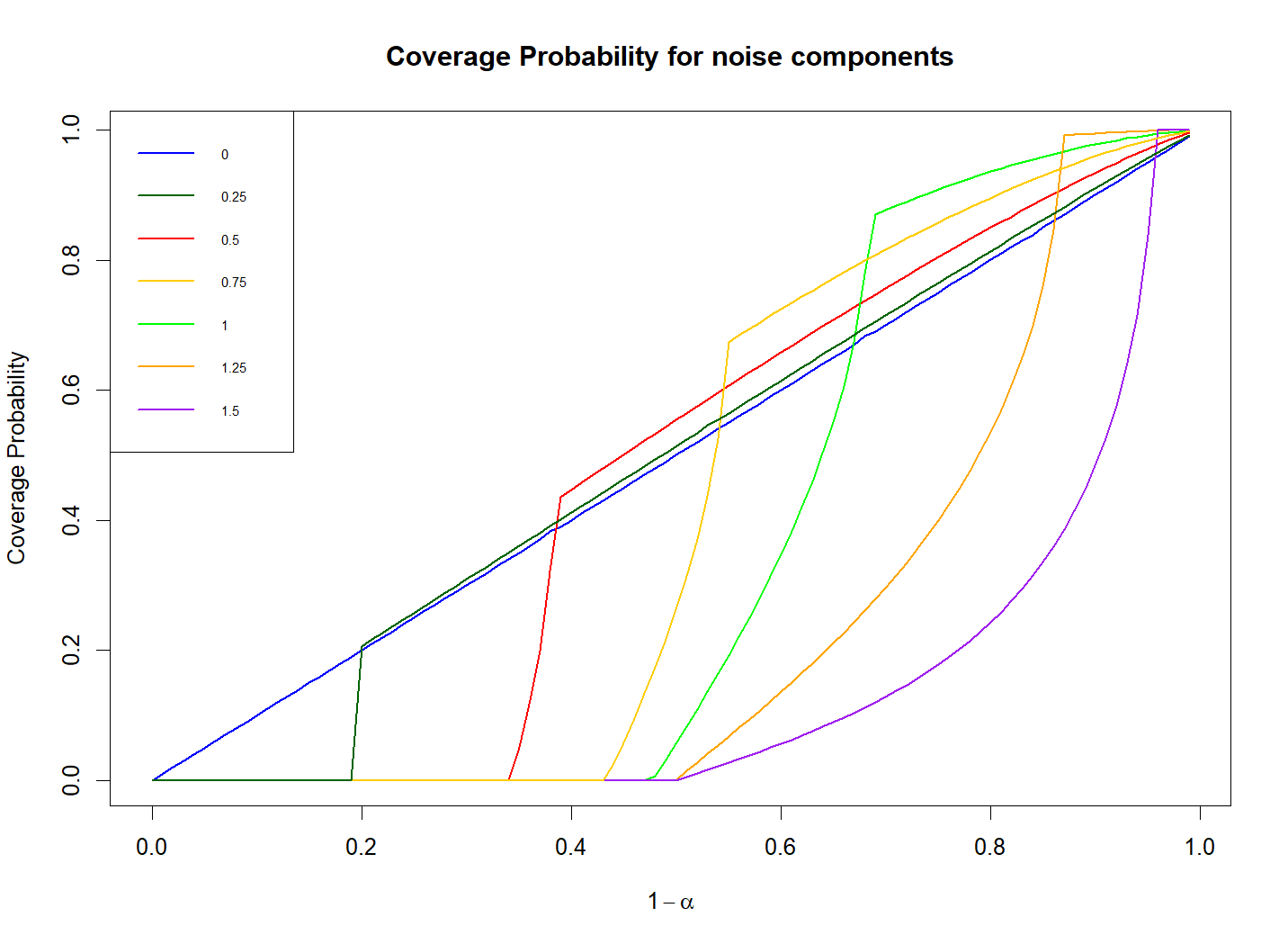}
    \caption{The upper panel shows the graph of $h_+(\lambda_0,\zeta)$ for the positive coefficients (\textit{left}), and that of $h(\lambda_0,\zeta)$ for the zero-valued components (\textit{right}) as a function of $\zeta$ for various values of $\lambda_0$. The lower panel shows coverage as a function of the credibility level for non-zero (\textit{left}) and zero (\textit{right}) components for various values  $\lambda_0$. The black diagonal line corresponds to the exact coverage.} 
    \label{fig:small_p_coverage}
\end{figure}

\begin{figure}[htbp]
    \centering
    \includegraphics[width = 0.7\linewidth, height = 0.5\linewidth]{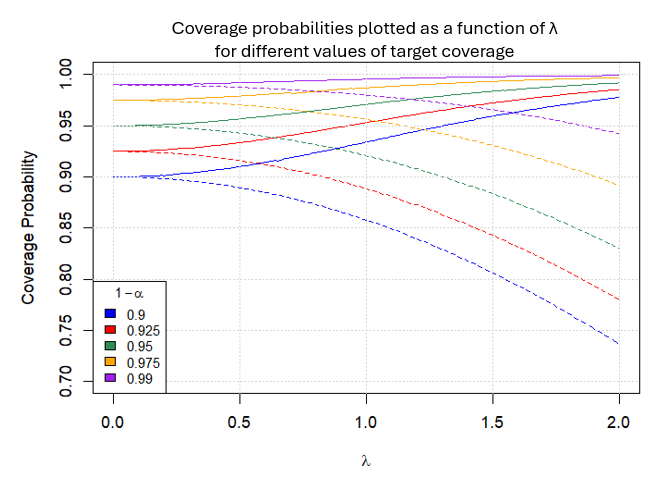}
    \caption{Plot of coverage probabilities against $\lambda$ for positive components (\textit{dashed lines}) and zero components (\textit{solid lines}) for a few interesting values of $1-\alpha$.}
    \label{fig:cov_vs_lam_for_diff_alphas}
\end{figure}

Based on the above discussion and the figures displayed above, we provide a guideline for choosing the confidence level for a wide range of coverage requirements. In \Cref{tab:my-tab_lambda0_guideline}, we give the value of $\gamma$ that attains the limiting coverage $1-\alpha$ for any penalty parameter between 0.05 to 4. For any value of the penalty parameter not listed in \Cref{tab:my-tab_lambda0_guideline}, the \texttt{solve\_gamma} function in the \texttt{credInt R} package \citep{credInt2024} may be used. Moreover, from the plot of $\psi_0(\alpha,\lambda)$ in \Cref{fig:small_p_coverage}, we can ensure higher than $(1-\alpha)$-coverage for the zero components for the chosen $(1-\gamma)$ level credible intervals. Thus, choosing $\gamma$ according to the table attains the limiting coverage exactly $1-\alpha$ for a non-zero parameter value and at least $1-\alpha$ for a zero value of the parameter.

\begin{table}[htbp]
\centering
\begin{tabular}{c|ccccc}
\hline
\backslashbox{$\lambda$}{$1-\alpha$} 
& 0.9    & 0.925  & 0.95   & 0.975  & 0.99   \\ \hline
0.05 & 0.9001 & 0.9251 & 0.9501 & 0.9751 & 0.9900 \\
0.1  & 0.9004 & 0.9254 & 0.9503 & 0.9752 & 0.9901 \\
0.15 & 0.9009 & 0.9258 & 0.9506 & 0.9754 & 0.9902 \\
0.2  & 0.9017 & 0.9264 & 0.9511 & 0.9757 & 0.9904 \\
0.25 & 0.9026 & 0.9272 & 0.9517 & 0.9761 & 0.9906 \\
0.3  & 0.9037 & 0.9282 & 0.9525 & 0.9765 & 0.9909 \\
0.35 & 0.9050 & 0.9293 & 0.9534 & 0.9771 & 0.9910 \\
0.4  & 0.9066 & 0.9306 & 0.9543 & 0.9777 & 0.9914 \\
0.45 & 0.9082 & 0.9320 & 0.9554 & 0.9784 & 0.9917 \\
0.5  & 0.9100 & 0.9335 & 0.9566 & 0.9790 & 0.9920 \\
0.55 & 0.9120 & 0.9352 & 0.9578 & 0.9798 & 0.9924 \\
0.6  & 0.9141 & 0.9369 & 0.9591 & 0.9806 & 0.9927 \\
0.65 & 0.9163 & 0.9387 & 0.9605 & 0.9814 & 0.9931 \\
0.7  & 0.9186 & 0.9406 & 0.9619 & 0.9822 & 0.9934 \\
0.75 & 0.9210 & 0.9426 & 0.9634 & 0.9830 & 0.9938 \\
0.8  & 0.9235 & 0.9446 & 0.9649 & 0.9838 & 0.9942 \\
0.85 & 0.9262 & 0.9468 & 0.9664 & 0.9846 & 0.9945 \\
0.9  & 0.9288 & 0.9488 & 0.9679 & 0.9854 & 0.9948 \\
0.95 & 0.9314 & 0.9510 & 0.9694 & 0.9862 & 0.9952 \\
1    & 0.9340 & 0.9530 & 0.9708 & 0.9871 & 0.9955 \\
1.1  & 0.9394 & 0.9572 & 0.9737 & 0.9885 & 0.9961 \\
1.2  & 0.9446 & 0.9613 & 0.9765 & 0.9899 & 0.9966 \\
1.3  & 0.9498 & 0.9652 & 0.9791 & 0.9912 & 0.9971 \\
1.4  & 0.9546 & 0.9688 & 0.9815 & 0.9923 & 0.9976 \\
1.5  & 0.9593 & 0.9722 & 0.9837 & 0.9934 & 0.9979 \\
1.6  & 0.9636 & 0.9754 & 0.9857 & 0.9943 & 0.9982 \\
1.7  & 0.9676 & 0.9783 & 0.9875 & 0.9951 & 0.9985 \\
1.8  & 0.9713 & 0.9810 & 0.9891 & 0.9958 & 0.9987 \\
1.9  & 0.9746 & 0.9833 & 0.9906 & 0.9964 & 0.9989 \\
2    & 0.9776 & 0.9854 & 0.9918 & 0.9959 & 0.9991 \\
2.2  & 0.9828 & 0.9889 & 0.9939 & 0.9978 & 0.9994 \\
2.4  & 0.9869 & 0.9917 & 0.9956 & 0.9984 & 0.9996 \\
2.6  & 0.9901 & 0.9939 & 0.9968 & 0.9988 & 0.9997 \\
2.8  & 0.9927 & 0.9955 & 0.9977 & 0.9992 & 0.9998 \\
3    & 0.9946 & 0.9967 & 0.9983 & 0.9994 & 0.9998 \\
3.5  & 0.9976 & 0.9986 & 0.9993 & 0.9998 & 0.9999 \\
4    & 0.9990 & 0.9994 & 0.9997 & 0.9999 & 1.000  \\ \hline
\end{tabular}
\caption{Calibration table for credibility level for an intended asymptotic coverage.}
\label{tab:my-tab_lambda0_guideline}
\end{table}

\begin{remark} \rm 
The component-wise credible intervals can be further used to get joint credible hyper-rectangle. Fix $k \leq p$ components. Suppose we want a $1-\beta$ credible region for those $k$ components. Then, we first build the individual $1-\alpha$ credible intervals for each of those $k$ chosen components using the sparse-projection posterior method with the corresponding credibility levels calculated from Table~\ref{tab:my-tab_lambda0_guideline}. Using the asymptotic independence resulting from the diagonal nature of the limiting covariance matrix, the asymptotic coverage of the joint $k-$dimensional hyper-rectangle thus formed is $(1-\alpha)^k$. Consequently, choosing $1-\alpha = (1-\beta)^{1/k}$, the targeted asymptotic coverage $1-\beta$ for the joint credible region can be achieved. 
\end{remark}

\section{Numerical Results}
\label{simulation}

We conduct a small-scale simulation to capture the performance of the credible sets obtained by projecting the conjugate posterior samples to the sparse subspace. We study several scenarios and report a complete numerical performance encompassing the dimension, the sparsity, and the signal strength of the underlying model. First, fixing the true sparsity at $s = 5$ and dimension at $p = 20$, we first generate data using an independent design matrix and signal strength $(-2,-1.5,0.5,1,2)$. We repeat this study with a correlated design matrix where an autoregressive process with a correlation of $0.7$ has been used. We study the coverage and length of the $(1-\gamma)$ credible intervals of the 5 signals and 5 randomly chosen noise variables. Here, $\gamma$ is chosen in accordance with \Cref{tab:my-tab_lambda0_guideline} to achieve our target coverage of 0.95. We will need the value of the penalty parameter used in the projection map to use the calibration information in \Cref{tab:my-tab_lambda0_guideline}. To specifically know what $\gamma$ value returns the credible intervals with exact $(1-\alpha)$ coverage, the \texttt{solve\_gamma} function within the \texttt{credInt R} package \citep{credInt2024} can be used. For all the simulations, we implement the LASSO regression of $\bY$ on $\bX$ and use the cross-validated penalty parameter in all the $R = 5000$ posterior maps. These cases are observed under three sample sizes, $n = 500,\,1000$, and $5000$. The \texttt{credInt} function returns the $(1-\gamma)$ credible intervals for all the regression coefficients. We repeat the simulations $200$ times and report the average coverage values. The goal is to see if the coverage of the component-wise credible intervals increases and attains the desired level with the increase in sample size. 

Besides this, we study the impact of sparsity on coverage for different significance levels. We vary $s_0$ from 5 to 30 and consider three levels, 90\%, 95\%, and 99\%, to see how coverage changes as sparsity increases and whether the results improve when the sample size increases from 500 to 1000. Note that the design matrix in this case is independent, and the error is normally distributed with error variance 1.

\begin{table}[htbp]
\centering
\resizebox{\linewidth}{!}{%
\begin{tabular}{ccccccccccccc}
\hline
Design & Model &  & $S_1$ & $S_2$ & $S_3$ & $S_4$ & $S_5$ & $N_1$ & $N_2$ & $N_3$ & $N_4$ & $N_5$ \\ \hline
\multirow{6}{*}{\rotatebox[origin=c]{90}{Independent}} & n = 500 & Coverage & 0.943 & 0.938 & 0.942 & 0.935 & 0.937 & 0.948 & 0.944 & 0.952 & 0.943 & 0.947 \\
 & p = 20 & Length & 0.112 & 0.108 & 0.119 & 0.110 & 0.111 & 0.113 & 0.109 & 0.112 & 0.117 & 0.116 \\ \cline{2-13} 
 & n = 1000 & Coverage & 0.949 & 0.945 & 0.947 & 0.951 & 0.948 & 0.955 & 0.953 & 0.954 & 0.950 & 0.957 \\
 & p = 20 & Length & 0.117 & 0.114 & 0.114 & 0.113 & 0.115 & 0.113 & 0.114 & 0.116 & 0.113 & 0.114 \\ \cline{2-13} 
 & n = 5000 & Coverage & 0.955 & 0.951 & 0.949 & 0.951 & 0.952 & 0.958 & 0.957 & 0.962 & 0.960 & 0.961 \\
 & p = 20 & Length & 0.114 & 0.119 & 0.110 & 0.116 & 0.115 & 0.115 & 0.116 & 0.118 & 0.119 & 0.116 \\ \hline
\multirow{6}{*}{\rotatebox[origin=c]{90}{Correlated}} & n = 500 & Coverage & 0.933 & 0.937 & 0.931 & 0.933 & 0.932 & 0.945 & 0.944 & 0.944 & 0.941 & 0.942 \\
 & p = 20 & Length & 0.119 & 0.121 & 0.117 & 0.118 & 0.119 & 0.119 & 0.115 & 0.116 & 0.118 & 0.120 \\ \cline{2-13} 
 & n = 1000 & Coverage & 0.943 & 0.948 & 0.939 & 0.942 & 0.942 & 0.956 & 0.954 & 0.953 & 0.950 & 0.953 \\
 & p = 20 & Length & 0.121 & 0.122 & 0.116 & 0.117 & 0.119 & 0.117 & 0.112 & 0.116 & 0.122 & 0.121 \\ \cline{2-13} 
 & n = 5000 & Coverage & 0.951 & 0.952 & 0.947 & 0.949 & 0.950 & 0.957 & 0.958 & 0.958 & 0.957 & 0.960 \\
 & p = 20 & Length & 0.118 & 0.117 & 0.115 & 0.120 & 0.118 & 0.130 & 0.127 & 0.123 & 0.119 & 0.125 \\ \hline
\end{tabular}%
}
\caption{Coverage and length of $(1-\gamma)$ credible interval of $5$ signals and randomly chosen $5$ noise variables studied over $3$ sample sizes, under the independent and correlated design settings with target coverage $0.95$.}
\label{tab:my-table_ind_and_corr}
\end{table}

\Cref{tab:my-table_ind_and_corr} provides the coverage and lengths of the component-wise intervals for the cases ($p = 20 $ and $ \btheta^0_{S_0} = (-2,-1.5,0.5,1,1.5)$) under the two setups, namely, independent design and correlated design. We observe that in both cases, there is a general pattern that coverage increases with an increase in sample size and almost attains the desired level, which is 95\% in this case. The numerical results thus align with the theoretical findings and provide enough support that the projection method leads to almost accurate asymptotic coverage. It is noteworthy that although the theoretical results for this method have only been proven under a limiting independence condition on the covariates, the method is capable of performing well even when some correlation is present between the explanatory variables in the model.

\Cref{fig:cover_vs_sparsity}, on the other hand, depicts the coverage property as the model becomes sparser. We see the general tendency is that the coverage slightly degrades as the number of truly active predictors decreases. However, the solid line (showing the coverage for $n = 1000$) stays uniformly above the dashed line (showing coverage for $n = 500$) for all three significance levels. This guarantees that the asymptotic full coverage will hold as the sample size grows.  Moreover, the sparse projection-posterior method is much more flexible in terms of coverage as the choice of the $\lambda_0$ values is user-dependent.

\begin{figure}[htbp]
    \centering
    \includegraphics[width = 0.6\linewidth]{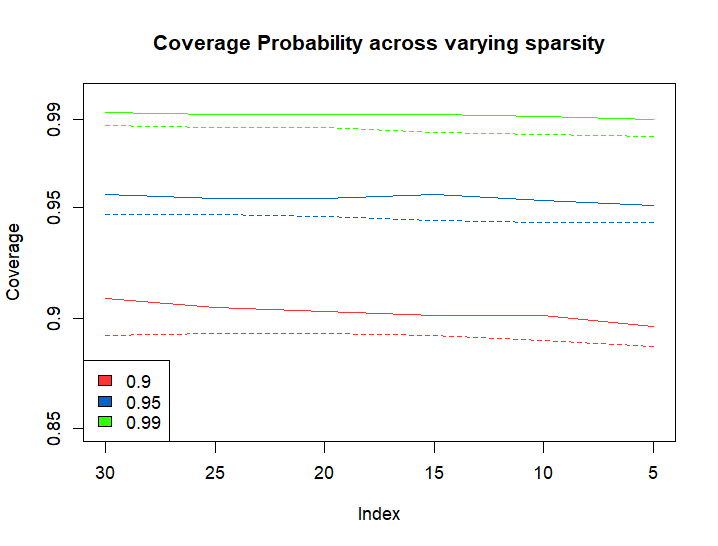}
    \caption{The average signal coverage of $\tau$-credible regions for three different levels $\tau = 0.9,0.95,0.99$ (averaged over the $s$ signals) are plotted against the degree of sparsity $s = 30,25,20,15,10,5$ with two choices of the sample size $n = 500, 1000$. The signal strength in all cases was $\theta^0_j = 1$ for all $j \in S = \{1,\dots,s\}$ and the results are based on $200$ replications.}
    \label{fig:cover_vs_sparsity}
\end{figure}

\begin{figure}[htbp]
    \centering
    \includegraphics[width = \linewidth]{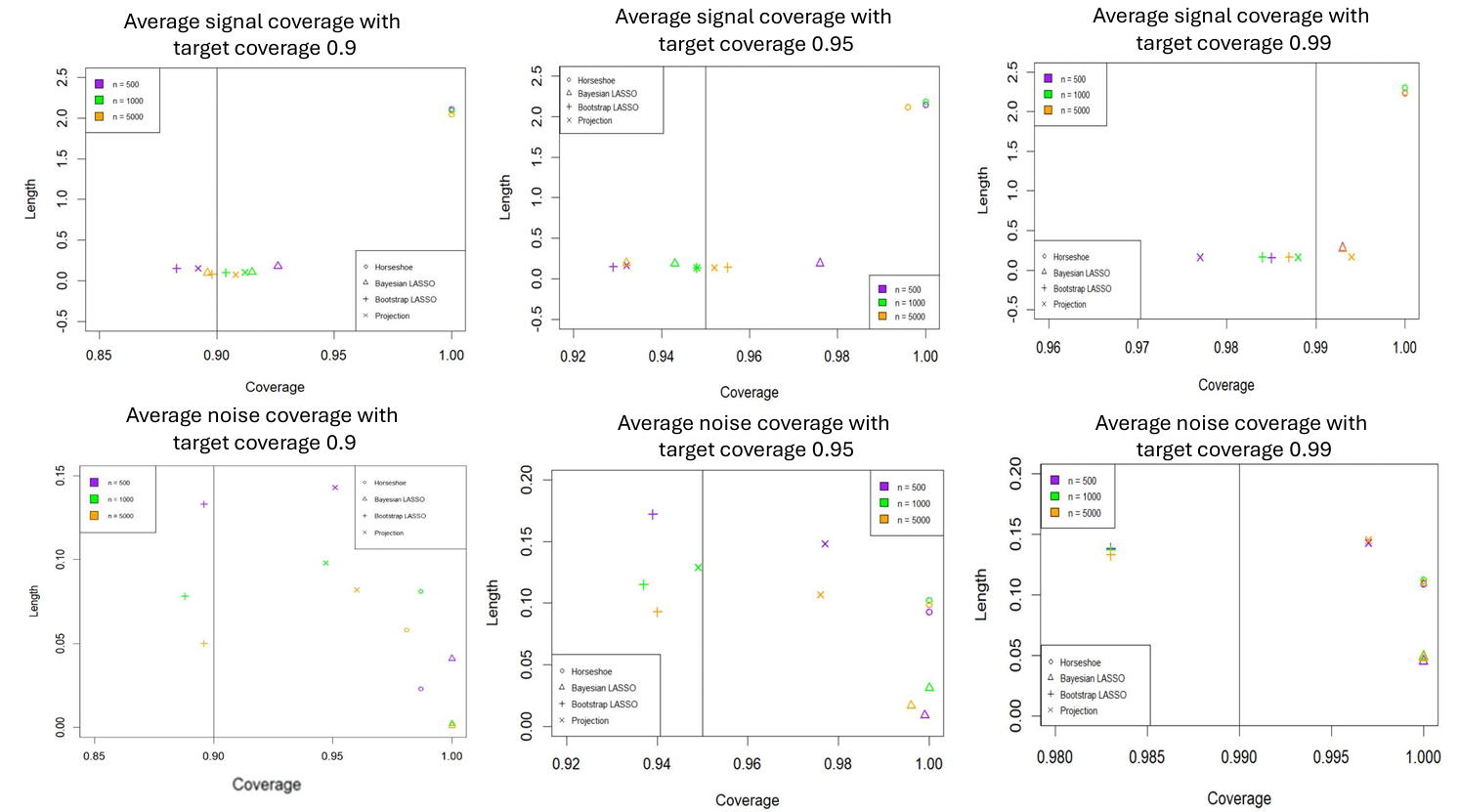}
    \caption{Comparison of Horseshoe, Bayesian LASSO, Bootstrap LASSO average signal (\textit{upper panel}) and noise (\textit{lower panel}) coverage with dimension $p = 50$, sparsity $s = 20$ and signal intensity $\beta_j^0 = 1.5$ for all $j = 1,2,\dots,s$, corresponding to three sample sizes, $n = 500,1000$ and $2000$.}
    \label{fig:p50_beta15_s20}
\end{figure}

\begin{figure}[htbp]
    \centering
    \includegraphics[width = \linewidth]{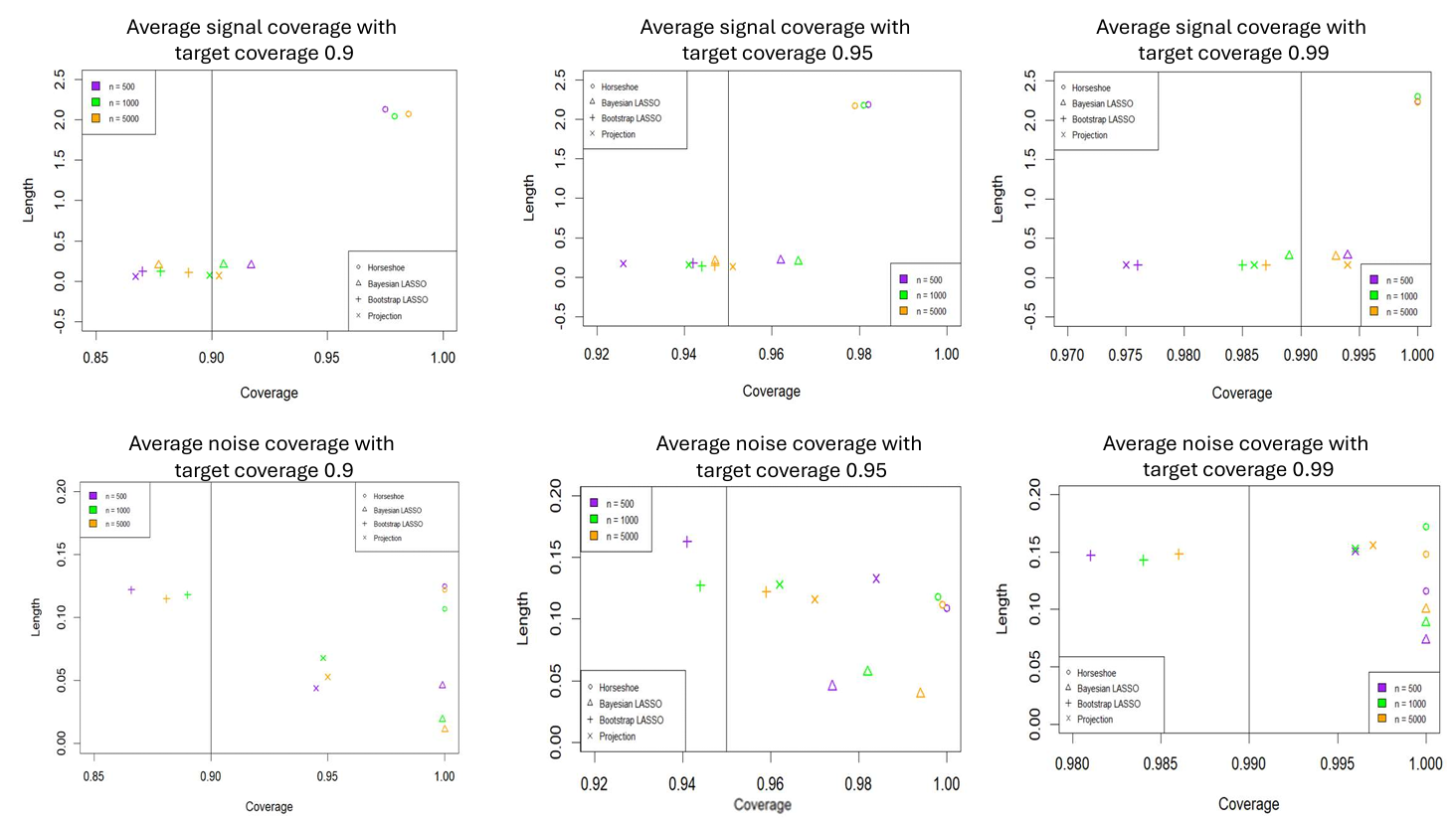}
    \caption{Comparison of Horseshoe, Bayesian LASSO, Bootstrap LASSO average signal (\textit{upper panel}) and noise (\textit{lower panel}) coverage with dimension $p = 100$, sparsity $s = 20$ and signal intensity $\beta_j^0 = 1.5$ for all $j = 1,2,\dots,s$, corresponding to three sample sizes, $n = 500,1000$ and $2000$.}
    \label{fig:p100_beta15_s20}
\end{figure}

Next, in Figures \ref{fig:p50_beta15_s20} and \ref{fig:p100_beta15_s20}, we compare the average signal and noise coverage and length of credible intervals for the Bayesian methods of Bayesian LASSO, Horseshoe, and the proposed sparse projection-posterior method, along with the frequentist method of Bootstrap LASSO. We repeat the study for three levels of significance $1-\alpha = 0.9, 0.95, 0.99$, and for two dimensions $p = 50,100$. We observe whether there is any change in the coverage over three sample sizes $n = 500,1000,5000$, keeping the sparsity and the signal intensity fixed at $s = 20$ and $\beta_j^0 = 1.5, \enskip j = 1,2,\dots,s$, respectively. Again, all the projections computed used the LASSO cross-validated penalty parameter. We drew $R = 5000$ posterior samples. The corresponding $\gamma$ was chosen from \Cref{tab:my-tab_lambda0_guideline}. Although the Bayesian LASSO does well in covering the true parameter values for low dimensions, there are two more aspects to consider. Firstly, the Bayesian LASSO does not have a model selection ability and always returns the full feature space. This causes a problem in identifying weaker signals. Secondly, the method is computationally infeasible when the dimension exceeds 100. On the other hand, the horseshoe method provides over-coverage for any choice of $1-\alpha$ at the high cost of yielding very long credible intervals. Accounting for these drawbacks, the sparse projection-posterior method is the only Bayesian method capable of providing exact coverage with reasonable lengths and comparable to the bootstrap LASSO technique of \cite{chatterjee2011bootstrapping}. Its coverage is sometimes even better than the bootstrap LASSO method, especially for larger sample sizes. Moreover, the bootstrap LASSO often fails to cover the noise variables, but the sparse-projection posterior method successfully provides coverage. As predicted by the theory, the cost is a slightly elongated interval. 

\section{Conclusion}\label{conclusion}

This study shows that the proposed Bayesian technique of projecting a dense conjugate posterior sample to a sparser domain performs well, offering adequate coverage by a modified credible under the fixed-dimensional setting. The method's ability to achieve reliable and consistent coverage holds promising implications for applications in various fields, as it successfully overcomes the drawbacks of existing methodologies to quantify the associated uncertainties in sparse regression. This technique may be extended to address the high dimensional setting and correlated predictors, and the method's robustness under different data distributions may be explored.

\begin{appendix}
\section*{Appendix: Proofs}\label{appn} 
\setcounter{equation}{0} 

\begin{lemma}\label{ridge CLT}
    Under Assumptions~\ref{predictor} and \ref{tuning}, and $a_n/\sqrt{n}\to 0$, we have that $\sqrt{n}(\hat{\btheta}^\mathrm{R} - \btheta^0) \rightsquigarrow \normal_p(\bm{0},\sigma_0^2\bC^{-1} )$. 
\end{lemma}

\begin{proof}
    If $\boldsymbol\varepsilon$ is normally distributed, representing  $$\hat{\btheta}^\mathrm{R} = (\bX^\mathrm{T}\bX + a_n \bm{I}_p)^{-1}\bX^\mathrm{T}\bX \btheta^0 + (\bX^\mathrm{T}\bX + a_n \bm{I}_p)^{-1}\bX^\mathrm{T} \bm{\varepsilon},$$ we observe that $\hat{\btheta}^\mathrm{R}$ is also normally distributed and hence it suffices to show that 
    \begin{align*} 
    \sqrt{n}(\mathbb{E}_0 (\hat\btheta^{\mathrm{R}})-\btheta^0) &=\sqrt{n} ((\bX^\mathrm{T}\bX + a_n \bm{I}_p)^{-1}\bX^\mathrm{T}\bX \btheta^0 -\btheta^0) 
    \to \bm{0}, \\
n   \mathbb{E}_0 [(\hat\btheta^{\mathrm{R}}-\btheta^0) (\hat\btheta^{\mathrm{R}}-\btheta^0)^{\mathrm{T}}] &=
    n\sigma_0^2(\bX^\mathrm{T}\bX + a_n \bm{I}_p)^{-1}\bX^\mathrm{T}\bX (\bX^\mathrm{T}\bX + a_n \bm{I}_p)^{-1} \\ 
    & \to \sigma_0^2 \bC^{-1}. 
    \end{align*}
    Since $n^{-1}\bX^\mathrm{T}\bX \to \bC$ and $a_n/\sqrt{n}\to 0$, we have $ (\bX^\mathrm{T}\bX + a_n \bm{I}_p)^{-1} \bX^\mathrm{T}\bX=\bm{I}_p+\bm{o}(n^{-1/2})$; here and elsewhere $\bm{o}(n^{-1/2})$ stands for a $p\times p$-matrix with all entries $o(n^{-1/2})$. The first relation now follows. The second relation holds even under the weaker condition $a_n/n\to 0$. 
    
   In general, when $\bm{\varepsilon}$ may not be normal, we verify, for any fixed non-zero linear combination $\bm{b}^{\mathrm{T}}\hat\btheta$, Lindeberg's condition for the central limit theorem. The condition reduces to 
   $(\max_{1\le i\le n} |\bm{b}^{\mathrm{T}} \bX_i|^2)/\sum_{i=1}^n  |\bm{b}^{\mathrm{T}} \bX_i|^2\to 0$, 
   owing to the identical distribution of the error variables, invoke the  Cram\'er-Wold device. 
   Since $$\max_{1\le i\le n}  n^{-1} |\bm{b}^{\mathrm{T}} \bX_i|^2\le \|\bm{b}\|^2 \max_{1\le i\le n}  n^{-1} |\bX_i\|^2\to 0, \; n^{-1}\sum_{i=1}^n |\bm{b}^{\mathrm{T}} \bX_i|^2 \to \bm{b}^{\mathrm{T}} \bC\bm{b}>0$$ by Assumption~\ref{predictor}, Lindeberg's condition is verified. 
  As $$\mathrm{Cov}(n^{-1/2} \sum_{i = 1}^n x_{ij}\varepsilon_i, n^{-1/2} \sum_{i = 1}^n x_{ik}\varepsilon_i) = n^{-1} \sigma_0^2 \sum_{i = 1}^n x_{ij}x_{ik} \to \sigma_0^2 c_{jk},$$ the multivariate central limit theorem applies to $(n^{-1/2} \sum_{i = 1}^n x_{ij}\varepsilon_i:1\le j\le p)$ with mean the zero vector and the dispersion matrix $\sigma_0^2 \bC$. Hence it follows that $(\bX^\mathrm{T}\bX + a_n \bm{I}_p)^{-1}\bX^\mathrm{T} \bm{\varepsilon}\rightsquigarrow \normal_p(\bm{0},\sigma_0^2 \bC^{-1})$.
\end{proof}

\begin{lemma}
    \label{error variance}
    Under Assumptions~\ref{predictor} and \ref{tuning}, and $a_n/\sqrt{n}\to 0$, the posterior distribution of $\sigma$ contracts at $\sigma_0$ at the rate $n^{-1/2}$, that is, for any $M_n\to \infty$
    \begin{align}
\label{posterior consistency sigma}
    \mathbb{E}_0 \Pi (|\sigma-\sigma_0|> M_n n^{-1/2}|\bY)\to 0.
\end{align}
\end{lemma}

\begin{proof}
It suffices to prove that the posterior for $\tau=\sigma^{-2}$ concentrates in $n^{-1/2}$-sized neighborhoods of $\tau_0=\sigma_0^{-2}$, which we verify from  \eqref{posterior sigma} using Chebyshev's inequality. We show that $\E(\tau|\bY)=\tau_0+O_p(n^{-1/2})$ and $\mathrm{Var}(\tau|\bY)=O_p(n^{-1})$. 
Let $\bm{\delta}_n=n^{-1}\bX^{\mathrm{T}}\bm{\varepsilon}$. Then, we can write $ (b_1+n/2)[\E(\tau|\bY)]^{-1}/(n/2)-2b_2/n$ as 
\begin{align*}
    (\btheta^0)^{\mathrm{T}} \bC_n \btheta^0+2\bm{\delta}_n^{\mathrm{T}} \btheta^0+ \|\bm{\varepsilon}\|^2/n -(\bC_n\btheta^0+\bm{\delta}_n)^{\mathrm{T}} (\bC_n +n^{-1}a_n \bm{I}_p)^{-1} (\bC_n\btheta^0+\bm{\delta}_n).
\end{align*}
This reduces to 
\begin{align*}
\sigma_0^2 &+O_p(n^{-1/2}) -(\btheta^0)^{\mathrm{T}} (\bC_n+ o(n^{-1/2})) 
    \btheta^0\\
    &\quad -2 \bm{\delta}_n^{\mathrm{T}} (\bm{I}_p+o(n^{-1/2}))\btheta^0- \bm{\delta}_n^{\mathrm{T}} (\bC^{-1} +o(n^{-1/2})) \bm{\delta}_n, 
\end{align*}
which is $\sigma_0^2 +O_p(n^{-1/2})$. This implies the first assertion. 

The second assertion follows because $ \bY^{\mathrm{T}} [\bm{I}_n-\bX(\bX^{\mathrm{T}}\bX+a_n \bm{I}_p)^{-1} \bX] \bY$ grows at the rate $n$ in $\mathbb{P}_0$-probability. 
\end{proof}

\begin{remark}\rm 
\label{root-n-consistency}
It is well-known that the least square estimator $\hat\btheta$ is $\sqrt{n}$-consistent for $\btheta$, and the variance estimator $\hat\sigma^2=n^{-1}\|\bY-\bX \hat\btheta\|^2$ is $\sqrt{n}$-consistent for $\sigma^2$. We can write $\hat\btheta^{\mathrm{R}}=(\bC_n+n^{-1}a_n \bm{I}_p)^{-1} \bC_n \hat\btheta=(\bm{I}_p+\bm{o}(n^{-1/2}))\hat\btheta$. Hence $\|\hat\btheta^{\mathrm{R}}-\hat\btheta\|=o_p(n^{-1/2})$, and in particular, $\hat\btheta^{\mathrm{R}}$ is $\sqrt{n}$-consistent for $\btheta$. Then it follows from the Cauchy-Schwarz inequality and Assumption~\ref{predictor} that 
$$n^{-1}\|\bY-\bX \hat\btheta^{\mathrm{R}}\|^2=n^{-1}[\|\bY-\bX \hat\btheta\|^2-2(\bY-\bX \hat\btheta)^{\mathrm{T}} \bX (\hat\btheta^{\mathrm{R}}-\hat\btheta)+\| \bX (\hat\btheta^{\mathrm{R}}-\hat\btheta)\|^2]$$   
is also $\sqrt{n}$-consistent for $\sigma^2$. 
\end{remark}

\begin{lemma}\label{pos_prob_at_0}
    Suppose $\theta_j^0 = 0$ for $j = s_0+1,\dots,p$. Then for any $\Delta$,  almost surely $$\P(\bm{T}^*_{S_0^c} =\bm{0}_{S_0^c})|\Delta) > 0.$$
\end{lemma}

\begin{proof}
    The minimizer of \eqref{limit distribution} can be expressed as $\bm{T}^* = \begin{bmatrix}
        \bm{T}_{S_0}^*\\
        \bm{T}_{S_0^c}^*
    \end{bmatrix}$.
    Write $$\bC =  \begin{bmatrix}
\bC_{11} & \bC_{12} \\
\bC_{21} & \bC_{22} 
\end{bmatrix} \textnormal{ and } \bm{W}^* = \begin{bmatrix}
        \bm{W}_{1}^*\\
        \bm{W}_{2}^*
    \end{bmatrix}.$$
    If $\bm{T}_{S_0^c}^* = \textbf{0}_{S_0^c} \in \R^{p-s_0}$, then, the KKT conditions of \eqref{limit distribution} can be written as 
    \begin{align}
    \label{KKT1}
        \bC_{11}\bm{T}_{S_0}^* - (\bC \bm{W}^*)_{1} = -\frac{\lambda_0}{2} \begin{bmatrix}
            \text{sign}(\theta_1^0)\\
            \vdots\\
            \text{sign}(\theta_{s_0}^0)
        \end{bmatrix}
    \end{align}
    and \begin{align}
    \label{KKT2}
        -\frac{\lambda_0}{2} \bm{1} \leq \bC_{21}\bm{T}_{S_0}^* - (\bC \bm{W}^*)_2 \leq \frac{\lambda_0}{2},
    \end{align} where $(\bC \bm{W}^*)_{1} = \bC_{11}\bm{W}^*_1 + \bC_{12}\bm{W}^*_2$, $(\bC \bm{W}^*)_{2} = \bC_{21}\bm{W}^*_1 + \bC_{22}\bm{W}^*_2$ and $\bm{1}$ is a $(p-s_0)$ dimensional vector of 1's. Then, solving for $\bm{T}_{S_0}^*$ from \eqref{KKT1} and substituting in \eqref{KKT2}, we have $$-\frac{\lambda_0}{2} \bm{1} \leq \bC_{21}\bC_{11}\big((\bC\bm{W}^*)_1 - \frac{\lambda_0}{2}\begin{bmatrix}
            \text{sign}(\theta_1^0)\\
            \vdots\\
            \text{sign}(\theta_{s_0}^0)
        \end{bmatrix}\big) - (\bC \bm{W}^*)_2 \leq \frac{\lambda_0}{2} \bm{1}.$$ Now, recalling that $\bm{W}^*|\bm\Delta \sim \normal(\sigma_0 \bm{C}^{-1/2}\bm{\Delta}, \sigma_0^2 \bm{C}^{-1})$, we have the result.
\end{proof}

\begin{proof}[Proof of \Cref{limit distribution}]
In view of Lemma~\ref{error variance}, it suffices to restrict $\sigma$ to a shrinking neighborhood $\mathcal{U}_n$ of $\sigma_0$ and derive the result by conditioning on $\sigma$ to a value $\sigma_n$, uniformly in $\sigma_n\in \mathcal{U}_n$. 
We can write $\btheta^* = \argmin_{\bm{u}} M_n(\bm{u})$, where 
\begin{align*}
    M_n(\bm{u}) = \argmin_{\bm{u}} \{(\bm{u} - \btheta)^\mathrm{T} \bX^\mathrm{T}\bX (\bm{u} - \btheta) + {\lambda_0}{\sqrt{n}} \|\bm{u}\|_1\}
\end{align*}
Then, defining $\bm{t} = \sqrt{n}(\bm{u} - \btheta^0)$, we have 
$$\bm{T}^*_n = \sqrt{n}({\boldsymbol{\theta}}^* - \btheta^0) = \argmin_{\bm{t}} \big\{ M_n ( \btheta^0 + n^{-1/2}{\bm{t}}) - M_n(\btheta^0) \big\}.$$ 
Because $\btheta|(\bY,\sigma=\sigma_n) \sim \normal(\hat{\btheta}^\mathrm{R}, \sigma_n^2 (\bX^\mathrm{T}\bX + a_n \bm{I}_p)^{-1})$, we can rewrite $\btheta = \hat{\btheta}^\mathrm{R} - \sigma_n (\bX^\mathrm{T}\bX + a_n \bm{I}_p)^{-1/2} \bm{U}_n^*$, where $\bm{U}_n^* \sim \normal_p(\bm{0},\bm{I}_p)$. Then, writing $\bm{\Delta}_n = -\sqrt{n}\bC_n^{1/2}(\hat{\btheta}^\mathrm{R}-\btheta^0)/\sigma_n$, we have 
\begin{eqnarray*}
    \lefteqn{M_n ( \btheta^0 + n^{-1/2} {\bm{t}}) - M_n(\btheta^0) } \nonumber\\
    & = &(\btheta^0 + n^{-1/2} {\bm{t}} - \btheta)^\mathrm{T} \bX^\mathrm{T} \bX(\btheta^0 + n^{-1/2} {\bm{t}} - \btheta) + \lambda_0 \sqrt{n} \|\btheta^0 + n^{-1/2}{\bm{t}}\|_1 \nonumber \\
    &  & - (\btheta^0 - \btheta)^\mathrm{T} \bX^\mathrm{T} \bX(\btheta^0 - \btheta) - \lambda_0 \sqrt{n} \|\btheta^0\|_1\nonumber\\
    & = &\bm{t}^\mathrm{T} \bC_n \bm{t} + 2 \bm{t}^\mathrm{T} \bC_n \sqrt{n}(\btheta - \btheta^0) + \lambda_0 \sqrt{n} ( \|\btheta^0 + n^{-1/2} {\bm{t}}\|_1 - \|\btheta^0\|_1)\nonumber\\
    & = &\bm{t}^\mathrm{T} \bC_n \bm{t} + 2 \bm{t}^\mathrm{T} \bC_n \sqrt{n}(\hat{\btheta}^\mathrm{R} - \sigma_n (n \bC_n + a_n\bm{I})^{-1/2}\bm{U}_n^* - \btheta^0) \nonumber \\ 
    &  & + \lambda_0 \sqrt{n} ( \|\btheta^0 + n^{-1/2} {\bm{t}}\|_1 - \|\btheta^0\|_1)\nonumber\\
    & = &\bm{t}^\mathrm{T}\bC_n \bm{t} - 2 \sigma_n \bm{t}^\mathrm{T} \bC_n^{1/2} \bm{\Delta}_n - 2 \sigma_n \bm{t}^\mathrm{T} \bC_n (n \bC_n + a_n\bm{I})^{-1/2} \bm{U}_n^* \nonumber \\
    & & + \lambda_0 \sqrt{n} ( \|\btheta^0 +n^{-1/2} {\bm{t}}\|_1 - \|\btheta^0\|_1). \nonumber
\end{eqnarray*}
In view  of \Cref{ridge CLT}, $\bm{\Delta}_n$ converges weakly to $\bm{\Delta}\sim \normal_p(\bm{0},\bm{I}_p)$. Writing $\bm{U}^*$ for a $\normal_p(\bm{0},\bm{I}_p)$ variable independent of $\bm{\Delta}$, it follows that the weak limit of  $M_n ( \btheta^0 + n^{-1/2}{\bm{v}}) - M_n(\btheta^0)$, conditional on the data and $\sigma=\sigma_n$, is the stochastic process 
\begin{align} 
     \bm{t}^\mathrm{T}\bC\bm{t} - 2 \sigma_0 \bm{t}^\mathrm{T}\bC^{1/2} (\bm{U}^*+\bm{\Delta}) + \lambda_0  \big[ \sum_{j=1}^{s_0} t_{j}\mathrm{sign}(\theta^0_{j}) + \sum_{j=s_0+1}^{p} |t_j| \big].
     \label{process limit}
\end{align}
The limiting process has a unique separated maximum. Hence, by the Argmax Theorem (Theorem~3.2.2 of \cite{van10weak}),  the conditional distribution of $\bm{T}_n^*$ given the data converges to the distribution of 
$$\bm{T}^*= \argmin_{\bm{t}} \{\bm{t}^\mathrm{T}\bC\bm{t} - 2 \sigma_0 \bm{t}^\mathrm{T}\bC^{1/2}(\bm{U}^*+\bm{\Delta}) + \lambda_0  \big[ \sum_{j=1}^{s_0} t_{j}\mathrm{sign}(\theta^0_{j}) + \sum_{j=s_0+1}^{p} |t_j| \big]\}$$ 
given $\bm{\Delta}$, in the space of probability measures $\mathfrak{M}(\mathbb{R}^p)$ under the weak topology. Writing $\bm{W}^*=\sigma_0 \bC^{-1/2} (\bm{U}^*+\bm{\Delta})$, the distributional representation in \eqref{projection_lim_rv_small_p} is obtained. 

Let $\Xi_n$ stand for the random probability measure $\mathcal{L}(\bm{T}_n^* \in \cdot|\bY, \sigma=\sigma_n)$. 
Writing $\Xi$ for the random probability measure $\mathcal{L}(\bm{T}\in \cdot |\bm\Delta)$, we have $\Xi_n \rightsquigarrow \Xi$ in the space  $\mathfrak{M}(\R^p)$ with respect to the topology of weak convergence.

We also note that the conditional distribution of the process $M_n ( \btheta^0 + n^{-1/2} {\bm{t}}) - M_n(\btheta^0)$ given the data and $\sigma=\sigma_n$ depends on the data through $\bm{\Delta}_n$, which is a continuous function of $\bX^{\mathrm{T}}\bm{\varepsilon}$. Also, $\sqrt{n}(\btheta^{\mathrm{L}}-\btheta_0)$ is a continuous function of $\bX^{\mathrm{T}}\bm{\varepsilon}$ too. Further $n^{-1/2} \bX^{\mathrm{T}}\bm{\varepsilon}\rightsquigarrow \bm{\Delta}$ on $\mathbb{R}^p$. 
Hence, combining with Theorem~\ref{lasso_lim_dist}, we can conclude that $(\Xi_n, \bm{\xi}_n)\rightsquigarrow (\Xi, \bm{\xi})$; that is, the joint weak convergence holds in the space $\mathfrak{M}(\mathbb{R}^p)\times \mathbb{R}^p$ with respect to the product of the weak topology and the Euclidean topology. 
\end{proof}

\begin{proof}[Proof of \Cref{coverage}]
Since a norm is a continuous functional, it follows from \eqref{centered joint distr} that 
\begin{align}
\label{joint_limit_for_coverage_result}
     \big(\mathcal{L}(\|\bm{T}_n^* - \hat{\bm\xi}_n\|_{\mathcal{K}} \in \cdot|\bY), \|\hat{\bm\xi}_n\|_{\mathcal{K}}\big) \rightsquigarrow \big(\mathcal{L}(\|\bm{T} -\bm\xi\|_{\mathcal{K}} \in \cdot|\bm{\Delta}), \|\xi\|_{\mathcal{K}} \big) 
\end{align}  
in the space $\mathfrak{M}(\R^p) \times \R^p$.  
To obtain the limiting coverage of the obtained credible region, note that  
$\{\btheta^0 \in B_r\}=\{\|\sqrt{n}(\btheta^0 - \hat{\btheta}^\mathrm{L}) \|_{\mathcal{K}} \leq r_{1-\alpha,n}\}$, 
and 
\begin{eqnarray*}
\lefteqn{\{\Pi(\|\sqrt{n}(\btheta^* - \hat{\btheta}^\mathrm{L})\|_{\mathcal{K}} \leq \|\sqrt{n}(\btheta^0 - \hat{\btheta}^\mathrm{L})\|_{\mathcal{K}}|\bY)< 1 - \alpha\} }\\ 
&& \subset \{\btheta^0 \in B_r\}\\
&& \subset \{\Pi(\|\sqrt{n}(\btheta^* - \hat{\btheta}^\mathrm{L})\|_{\mathcal{K}} \leq \|\sqrt{n}(\btheta^0 - \hat{\btheta}^\mathrm{L})\|_{\mathcal{K}}|\bY)\le 1 - \alpha\}. 
\end{eqnarray*}
The assertion is now immediate.  
\end{proof}

\begin{proof}[Proof of \Cref{coverage of credible interval}]
    Let the limit $\bC$ of $n^{-1}\bX^\mathrm{T}\bX$ be a block-diagonal matrix with blocks $\{j\}$ and $[-j]$. As argued in Remark~\ref{scaling}, it suffices to prove the result under the additional assumption that $c_j =1$ and $\sigma_0 =1$. Then $\xi_j$ reduces to 
    $\argmin\{ v_j^2 - 2 \Delta_jv_j + \lambda_0 [  v_j \text{sign}(\theta^0_j) + \mathbbm{1}(\theta_j^0=0) |v_j|]\}$ 
    For a non-zero-component, the minimizer of $v_j^2 - 2\Delta_jv_j + \lambda_0 v_j \text{sign}(\theta^0_j)$ is  
    \begin{align}
    \label{non-zero}
    \xi_j = \begin{cases}
    \Delta_j - {\lambda_0}/{2}, & \text{if } \theta_j^0 > 0,\\
    \Delta_j +  {\lambda_0}/{2}, & \text{if } \theta_j^0 < 0,
\end{cases}
\end{align} 
while for a zero-component, the minimizer of $v_j^2 - 2\Delta_jv_j + \lambda_0 |v_j|$ is  \begin{align}
\label{zero}
    \xi_j = \begin{cases}
    \Delta_j - {\lambda_0}/{2}, & \text{if } \Delta_j >  {\lambda_0}/{2},\\
    \Delta_j +  {\lambda_0}/{2}, & \text{if } \Delta_j < - {\lambda_0}/{2},\\
    0, & \text{if } |\Delta_j| \leq  {\lambda_0}/{2}.
\end{cases}
\end{align} 
In both cases, we note that $\xi_j$ is a continuous function of $\Delta_j$. In the first case, $\xi_j$ is a continuous random variable, while in the second case, $\xi_j$ has a point mass at $0$. 
Proceeding in the same way, we have for a  non-zero component 
\begin{align}
\label{T_j_star_non_zero}
    T_j^*  = \argmin \{t_j^2 - 2 W_j^* t_j + \lambda_0 t_j \textnormal{sign}(\theta_j^0)\} = \begin{cases}
        W_j^* -  {\lambda_0}/{2}, & \text{if } \theta_j^0 > 0,\\
        W_j^* +  {\lambda_0}/{2}, & \text{if } \theta_j^0 < 0.
    \end{cases}
\end{align}
and for a zero-component 
\begin{align}
\label{T_j_star_zero}
    T_j^*  = \argmin \{t_j^2 - 2 W_j^* t_j + \lambda_0 |t_j|\} 
    = \begin{cases}
        W_j^* -  {\lambda_0}/{2}, & \text{if } W_j^* >  {\lambda_0}/{2},\\
        W_j^* +  {\lambda_0}/{2}, & \text{if } W_j^* < - {\lambda_0}/{2},\\
        0, & \text{if } |W_j^*| \leq  {\lambda_0}/{2},
    \end{cases}
\end{align}
where $W_j^*|\Delta_j\sim \normal (\Delta_j,1)$. 
Writing $Z_j = W_j^* - \Delta_j$, and noting that $Z_j|\Delta_j \sim \normal(0,1)$, 
we obtain that for $\theta_j^0 > 0$, 
\begin{align} 
\label{positive case}
\P(|T_j^* - \xi_j| \leq |\xi_j| \big| \Delta_j)=\P(|Z_j| \leq |\xi_j| \big| \Delta_j )=
2 \Phi(|\xi_j|) - 1=h_+(\lambda_0,\Delta_j),
\end{align}
because $T_j^*-\xi_j=(W_j^*-\lambda_0/2)-(\Delta_j-\lambda_0/2)=Z_j$ in this case. 

Similarly, for $\theta_j<0$, $T_j^*-\xi_j=(W_j^*+\lambda_0/2)-(\Delta_j+\lambda_0/2)=Z_j$ again, and hence 
\begin{align}
\label{negative case}
\P(|T_j^* - \xi_j| \leq |\xi_j| \big| \Delta_j)=\P(|Z_j| \leq |\xi_j| \big| \Delta_j )=
2 \Phi(|\xi_j|) - 1=h_-(\lambda_0,\Delta_j).
\end{align}
Since $2 \Phi(|\xi_j|) - 1 $ is a continuous random variable, the coverage probability $\prob_0 (\theta_j^0\in [\hat\theta_j^{\mathrm{L}}-r_{1-\alpha,j},\hat\theta_j^{\mathrm{L}}+r_{1-\alpha,j}])$ converges to 
\begin{align*}
    \P ( 2 \Phi(|\xi_j|) - 1 \leq 1 - \alpha)
    =  \P(|\xi_j| \leq z_{\alpha/2})
    = \Phi(\lambda_0/2 + z_{\alpha/2}) - \Phi(\lambda_0/2 - z_{\alpha/2})
\end{align*}
whenever $\theta^0_j\ne 0$.

To compute $\P(|T_j^* - \xi_j| \leq |\xi_j| \big| \Delta_j)$ for $\theta_j^0 = 0$, we consider the three cases separately. For $\Delta_j > \lambda_0/2$, we have $\xi_j=\Delta_j-\lambda_0/2>0$ so that from \eqref{T_j_star_zero} and \eqref{zero}, 
\begin{align*}
    \P(|T_j^* - \xi_j| \leq |\xi_j| \big| \Delta_j)
    & = \P(|W_j^* - \lambda_0/2 - \xi_j| \leq \xi_j, W_j^* > \lambda_0/2 \big| \Delta_j) \\
    & \quad + \P(|W_j^* + \lambda_0/2 - \xi_j| \leq \xi_j, W_j^* < -\lambda_0/2 \big| \Delta_j) \\
    & \quad + \P(|0 - \xi_j| \leq \xi_j, |W_j^*| \leq \lambda_0/2 \big| \Delta_j)\\
    & = \P(|Z_j| \leq \Delta_j - \lambda_0/2, Z_j > -\Delta_j + \lambda_0/2 \big| \Delta_j)\\
    & \quad + \P(|Z_j + \lambda_0| \leq \Delta_j - \lambda_0/2, Z_j < -\Delta_j - \lambda_0/2 \big| \Delta_j)\\
    & \quad + \P(|-\Delta_j + \lambda_0/2| \leq \Delta_j - \lambda_0/2, |Z_j + \Delta_j| \leq \lambda_0/2 \big| \Delta_j)\\
    & = \P(-\Delta_j + \lambda_0/2 \leq Z_j \leq \Delta_j - \lambda_0/2| \Delta_j) + 0 \\
    & \quad + \P( - \Delta_j - \lambda_0/2 \leq Z_j \leq -\Delta_j + \lambda_0/2|\Delta_j)\\
    & = \Phi(\Delta_j - \lambda_0/2) - \Phi(-\Delta_j - \lambda_0/2)
\end{align*}
For the case where $\Delta_j < -\lambda_0/2$, we have $|\xi_j| = -\xi_j$, so that 
\begin{align*}
   \P(|T_j^* - \xi_j| \leq |\xi_j| \big| \Delta_j)
    & = \P(|Z_j - \lambda_0| \leq -\Delta_j - \lambda_0/2, Z_j > -\Delta_j + \lambda_0/2\big| \Delta_j)\\
    & \quad + \P(|Z_j| \leq -\Delta_j - \lambda_0/2, Z_j < - \Delta_j - \lambda_0/2\big|\Delta_j)\\
    & \quad + \P(|\Delta_j + \lambda_0/2| \leq |\Delta_j + \lambda_0/2|, |Z_j + \Delta_j| \leq \lambda_0/2\big| \Delta_j)\\
    & = \Phi(-\Delta_j + \lambda_0/2) - \Phi(\Delta_j + \lambda_0/2).
\end{align*}
Finally, when $|\Delta_j| \leq \lambda_0/2$, $\xi_j=0$ so that 
\begin{align*}
   \P(|T_j^* - \xi_j| \leq |\xi_j| \big| \Delta_j)
    &= \P(T_j^*=0\big|\Delta_j)\\ 
    &=  \P(|Z_j+\Delta_j|\le \lambda_0/2\big|\Delta_j)\\
    & = \Phi(\Delta_j + \lambda_0/2) - \Phi(\Delta_j - \lambda_0/2). 
\end{align*}
Thus in all cases for $\theta_j^0=0$, $ \P(|T_j^* - \xi_j| \leq |\xi_j| \big| \Delta_j)=h_0(\lambda_0,\Delta_j)$, which is a continuous random variable because $\zeta \mapsto h_0(\lambda_0,\zeta)$ is continuous and not flat, and $\Delta_j\sim \normal(0,1)$ is continuous. Therefore, it follows that 
\begin{align*} 
\prob_0 (\theta_j^0\in [\hat\theta_j^{\mathrm{L}}-r_{1-\alpha,j},\hat\theta_j^{\mathrm{L}}+r_{1-\alpha,j}])\to 
\P (h_0(\lambda_0,\Delta_j)\le 1-\alpha)=\psi_0(\alpha,\lambda_0),
\end{align*}
when $\theta_j^0=0$. 
\end{proof}

\end{appendix}

\section{Funding}
Research is partially supported in part by ARO grant number 76643MA 2020-0945.

\bibliography{refer}
\end{document}